\documentclass[twocolumn, 10pt]{article}

\usepackage{float}
\usepackage{blindtext}
\usepackage{hyperref}
\usepackage{comment}
\usepackage{placeins}
\usepackage{amsthm}
\usepackage{cite}
\usepackage{amsfonts}
\usepackage[margin=0.8in]{geometry}
\usepackage[utf8]{inputenc}
\usepackage{titlesec}
\usepackage{sectsty}
\usepackage{color, enumitem}
\usepackage{algorithm}
\usepackage[noend]{algorithmic}
\usepackage{amsmath}
\usepackage{amsthm}
\usepackage{xcolor}
\usepackage{caption}
\usepackage{graphicx}
\usepackage{verbatim}
\usepackage{xspace,amsfonts,amssymb,tikz,multirow,stmaryrd,soul}
\usepackage{mathtools, nccmath}

\newcommand{\moe}{\ensuremath{\textit{MoE}}\xspace}

\theoremstyle{definition}
\usepackage{amsthm}
    \newtheorem{theorem}{Theorem}
    \theoremstyle{definition}
    \newtheorem{definition}[theorem]{Definition}
    
    \newtheorem{lemma}[theorem]{Lemma}

\newcommand{\Lap}{\ensuremath{\textsf{Lap}}\xspace}

\newcommand{\lap}{\Lap}

\newcommand{\sucht}{\text{ st }}

\newcommand{\algn}[1]{\ensuremath{\text{\sf #1}}\xspace}

\newcommand{\namedref}[2]{\hyperref[#2]{#1}}
\newcommand{\subfigureref}[2]{\hyperref[#1]{Figure~\ref*{#1}#2}}

\definecolor{darkred}{rgb}{0.5, 0, 0}
\definecolor{darkgreen}{rgb}{0, 0.5, 0}
\definecolor{darkblue}{rgb}{0,0,0.5}

\newcommand{\eps}{\ensuremath{\varepsilon}}
\newcommand{\edp}[1]{\ensuremath{#1}-differentially private}
\newcommand{\Db}{\ensuremath{X}\xspace}  %random variable database
\newcommand{\db}{\ensuremath{x}\xspace}  %specific database
\newcommand{\dbu}{\ensuremath{\mathcal{X}}\xspace}  %\database universe

\newcommand{\xmin}{\ensuremath{x_{\text{min}}}\xspace}
\newcommand{\xmax}{\ensuremath{x_{\text{max}}}\xspace}

\newcommand{\bxmj}{\bar{\db}_{\textrm{-}j}}

\newcommand{\cilap}{\namedref{\algn{NOISYVAR}}{alg:Lap}\xspace}
\newcommand{\abslap}{\namedref{\algn{NOISYMAD}}{alg:AbsLap}\xspace}
\newcommand{\cci}{\namedref{\algn{SIM}}{alg:CCI}\xspace}
\newcommand{\cexp}{\namedref{\algn{CENQ}}{alg:ExpM}\xspace}
\newcommand{\cdbl}{\namedref{\algn{SYMQ}}{alg:DblM}\xspace}
\newcommand{\expmed}{\namedref{\algn{EXPQ}}{alg:ExpMed}\xspace}
\newcommand{\mad}{\namedref{\algn{MOD}}{alg:MAD}\xspace}

\newcommand{\qnorm}{\ensuremath{q_Z}}
\newcommand{\qt}{\textit{qt}}
\newcommand{\drange}{(\xmax - \xmin)}

\let\svthefootnote\thefootnote
\newcommand\blankfootnote[1]{%
  \let\thefootnote\relax\footnotetext{#1}%
  \let\thefootnote\svthefootnote%
}

\newcounter{ag}
\addtocounter{ag}{1}

\newcounter{ab}
\addtocounter{ab}{1}

\DeclareMathSizes{10}{9}{7}{5}
\titleformat*{\section}{\LARGE\bfseries}
\titleformat*{\subsection}{\large\bfseries}
\titleformat*{\subsubsection}{\normalsize\bfseries}
\titleformat*{\paragraph}{\normalsize\bfseries}
\titleformat*{\subparagraph}{\normalsize\bfseries}

\begin{document}

\title{\LARGE \textbf{Differentially Private Confidence Intervals}}

\author{\begin{tabular}{c c c c c}
Wenxin Du & & Canyon Foot & & Monica Moniot \\
\textit{\small wendu@reed.edu} & & \textit{\small canafoot@reed.edu} & & \textit{\small mamoniot@reed.edu} \\
& Andrew Bray &  & Adam Groce & \\
& \textit{\small abray@reed.edu} & & \textit{\small agroce@reed.edu} &  \\[10pt]
\multicolumn{5}{c}{Reed College Mathematics Department} \\
\multicolumn{5}{c}{Portland, OR, USA} \\
\end{tabular}}

\date{}

\maketitle{}

\begin{abstract}
{Confidence intervals for the population mean of normally distributed data are some of the most standard statistical outputs one might want from a database.  In this work we give practical differentially private algorithms for this task.  We provide five algorithms and then compare them to each other and to prior work.  We give concrete, experimental analysis of their accuracy and find that our algorithms provide much more accurate confidence intervals than prior work.  For example, in one setting (with $\epsilon=0.1$ and $n=2782$) our algorithm yields an interval that is only $1/15^\text{th}$ the size of the standard set by prior work.}
\end{abstract}

\section{Introduction}\label{sec:intro}

Estimating the mean of a population is one of the most basic tasks in statistics.  A medical researcher who wants to know the average height of an adult male would generally get an estimate by measuring the height of a random sample of people.  But when this value is reported, statisticians are usually careful not to just report a single point estimate.  They instead include some measure of the \textit{uncertainty} of this estimate.  That is, what is the range in which the true average of the population might plausibly fall?  This range is given as a \textit{confidence interval}.

In classical statistics, confidence intervals are usually easy to compute.  It is often acceptable to assume that continuous variables are (approximately) normally distributed.  When estimating the mean of a normally distributed population, given a random sample, the most accurate confidence interval is given by a simple calculation on the sample mean and sample variance.  This calculation has been known at least since the codification of confidence intervals in 1937 \cite{neyman1937x}.  However, things become more complicated when confidence intervals are being computed under the constraints of differential privacy.

When computing a confidence interval with differential privacy, there are two sources of randomness to consider.  As in the public setting, there is the sampling variability that arises when selecting a random sample of data. Private algorithms introduce a second source of randomness in order to preserve privacy. There has been surprisingly little work on the computation of private confidence intervals.  Karwa and Vadhan \cite{Vadhan2017} study the problem in depth from a theoretical standpoint, finding an algorithm with very good asymptotic performance, but poor practical performance.  They state that ``designing practical differentially private algorithms for confidence intervals remains an important open problem, whose solution could have wide applicability.''  It is precisely this open problem that we are attempting to solve.

Our main contribution is a set of five new, differentially private algorithms that output a confidence interval for the population mean of normally distributed data.  We do not assume that the population variance is known.  Two of our algorithms use Laplace noise, while three others use an exponential mechanism-based algorithm to report quantiles of the data.

All of our algorithms are experimentally verified to confirm that the resulting intervals are valid.  (I.e., they cover the true population mean with the desired frequency.)  We also experimentally compare our algorithms to the existing work on this question.  We find that our best algorithms consistently outperform prior work, often by large margins.  For example, with $\eps=0.1$, a range of [-32, 32], and a sample of size 2782, our best algorithm gives an interval width approximately 2.43 times that obtained without privacy, while the best prior work algorithm gives an interval that is about 37.10 times as wide.  This means that the cost of privacy has been reduced by 96\%.

We note also that many tasks in private data analysis could probably be made more accurate by assuming some information about the distribution of the underlying data.  Our quantile algorithm, for example, when applied to the median, gives a better estimate of the mean of a normal distribution than does the standard Laplace noise-based estimate.  Given that many statistical analyses being performed on private data \textit{already} make these assumptions about its distribution, we see no reason why privacy researchers should not measure utility under these assumptions.  (We stress that our privacy guarantees do not depend on any assumptions about the data.)

Our algorithms are all implemented and code is publicly available at \url{https://github.com/wxindu/dp-conf-int}.

\section{Background}\label{sec:bkg}

Below we first describe differential privacy, a well-established security definition for the private release of statistical queries on sensitive databases.  We then discuss the particular sort of query that we study --- confidence intervals for the mean of normally distributed data.  Finally, we discuss how privacy interacts with the goal of accurate confidence intervals and related prior work.

\subsection{Differential Privacy}

We imagine an analyst who issues queries to a database of private information.  The analyst might be untrusted, or they might be trusted but wish to release the query results publicly.  Either way it must be guaranteed that the output of the query protects the privacy of the individuals whose data is contained in the database.  Differential privacy \cite{dwork2006calibrating} formalizes such a guarantee.  Intuitively, the guarantee given by differential privacy to an individual is that any output will be roughly equally likely regardless of what data that individual submitted to the database.  This implies that no adversary could infer anything about an individual as a result of their participation in the database.  (This interpretation is subtle.  Interested readers should see \cite{kasiviswanathan2008note} or \cite{bassily2013coupled} for more discussion.)

To formalize this notion, we first define \textit{neighboring} databases.

\begin{definition}[Neighboring Databases]
Databases $\db,\db'\in \dbu$ are neighbors if one can be transformed to the other by changing the value of a single row $\db_i$.
\end{definition}

We can now define differential privacy:

\begin{definition}[Differential Privacy]
A query $f$ is $(\eps, \delta)$-differentially private if for all neighboring $\db, \db'\in \dbu$, and for all sets $S$ of possible outputs,
$$\Pr[f(\db) \in S] \leq e^\eps \Pr[f(\db') \in S] + \delta.$$
\end{definition}

This definition, sometimes called \textit{bounded} differential privacy \cite{kifer2011no}, is one of two common variants.  The other (\textit{unbounded}) defines neighboring databases as having a row deleted rather than changed.  The only significant difference is whether the size of the database must be protected.  Our algorithms achieve privacy with $\delta=0$, so from here on we state theorems only for the $\delta=0$ case, though the $\delta>0$ case is always similar.

The value $\eps > 0$ is considered the \emph{privacy parameter}. Smaller values correspond to less information being revealed about individuals in the database, thus stronger privacy. A value of $\eps = 1$ is fairly high but still meaningfully protects privacy, while $\eps = .01$ is quite low and allows many more queries on the database to be released while still maintaining a strong privacy guarantee, but it also requires queries to be less accurate.  The choice of $\eps$ is a policy decision.

There are two particularly useful properties of differential privacy that warrant mention.  The first, resistance to post-processing, requires that anything computed from private output is itself private.  This is a necessary feature of a good privacy definition, but it is also a useful tool that allows the easy construction of private queries.

\begin{theorem}[Post-processing \cite{dwork2006calibrating}]
For an $\eps$-differentially private query $f$, and any function $g$, the query $g \circ f$ is also $\eps$-differentially private.
\end{theorem}

Next we give the standard composition theorem, which shows that private queries can be combined in an acceptable way.

\begin{theorem}[Composition \cite{dwork2006calibrating}]
If query $f_1$ is $\eps_1$-differentially private and query $f_2$ is $\eps_2$-differentially private, their composition $f(\db)=(f_1(\db), f_2(\db))$ is $(\eps_1 + \eps_2)$-differentially private.
\end{theorem}

Composition allow for the idea of a \textit{privacy budget}, a total $\eps$ value that can be divided up as an analyst wishes between any number of queries.  It also allows complex queries to be constructed from several smaller, simpler queries.

Finally, we present two general methods for creating private queries.  The first is the Laplace mechanism, which adds noise to a query to create a private version.  This noise must be proportional the \textit{sensitivity} of the non-private query, defined as the maximum effect a single row can have on the output.

\begin{definition}[Global Sensitivity]
The sensitivity of a function $f: \dbu \rightarrow \mathbb{R}$, abbreviated $\Delta f$, is defined as
$$\Delta f = \max_{\substack{\db, \db'\in \dbu\text{ that}\\\text{are neighbors}}} |f(\db) -f(\db')|.$$
\end{definition}

The noise added is taken from a Laplace distribution, defined below.

\begin{definition}[Laplace Distribution]\label{def:ld}
The Laplace Distribution centered at 0 with scale $b$ is the distribution with probability density function:
$$ pdf(z)=\frac{1}{2b}\textup{exp}\Big(-\frac{|z|}{b}\Big).$$
We write $\lap(b)$ to denote the Laplace distribution with scale $b$.
\end{definition} 

This now allows a definition of the Laplace mechanism, the most standard generic technique for privatizing a given query.

\begin{theorem}[Laplace mechanism \cite{dwork2006calibrating}]
Given any query $f: \dbu \rightarrow \mathbb{R}$, the query
$$\widetilde{f}(\db) = f(\db) + L\text{ where }L \sim \lap\left(\frac{\Delta f}{\eps}\right)$$
is \edp{\eps}. 
\end{theorem}

In many cases the Laplace mechanism may not be appropriate. For instance, the output of the query might be categorical rather than numerical, or the query might be numerical but does not have the property that values close to one another have similar usefulness. In these cases, the \emph{exponential mechanism} can be a useful alternative. The exponential mechanism relies on a utility function $u:\dbu \times R \rightarrow \mathbb{R}$ which assigns a utility to any particular query response, given a particular database. ($R$ is the set of possible query outputs.) The utility function has a sensitivity as well, defined as
$$\Delta u = \max_{\substack{\db, \db'\in \dbu\text{ neighbors}\\r\in R}} |u(\db,r) - u(\db',r)|.$$
The exponential mechanism selects higher-utility outputs more often, with the probability of a given output increasing exponentially with its utility.

\begin{theorem}[Exponential Mechanism \cite{mcsherry2007mechanism}]
Given a query $f$ and a utility function $u$, define $\widetilde{f}$ such that $\forall\,\db\in\dbu,\forall\, r\in R,$
$$\Pr[\widetilde{f}(\db) = r] \propto \exp\left(\frac{\eps u(\db, r)}{2 \Delta u}\right).$$
Then $\widetilde{f}$ is $\eps$-differentially private. 
\end{theorem}

The exponential mechanism is not necessarily efficiently computable, but it can be for particular queries/utilities.

\subsection{Confidence Intervals}

In statistical inference, an analyst seeks to use a particular sample of data to infer attributes of the larger population from which it was drawn.  One common goal is to estimate a population parameter $\theta$.  A confidence interval algorithm takes as input a sample and outputs a range in which the population parameter is likely to fall.  Note that we are now considering the database \Db as a random variable.

\begin{definition}[Confidence intervals]
Given a database $X = (X_1, ..., X_n)$ of i.i.d.~samples from a population, a confidence interval algorithm $c$ outputs a closed interval $[a,b]\in \mathbb{R}$.  An algorithm with confidence level $(1 - \alpha)$, for $\alpha \in [0, 1]$, has the property that
$$\Pr[\theta \in c(X)] \geq 1 - \alpha.$$ 
\end{definition}
The probability in the above definition is traditionally taken over the randomness of the sample $X$, but it could (and in the private case will) also be taken over the randomness of $c$, were $c$ to be a randomized algorithm.  This probability, $\Pr[\theta \in c(X)]$, is called the \textit{coverage} of the confidence interval algorithm.

Of course, one could construct an interval which is guaranteed to contain the true value by releasing all of $\mathbb{R}$, but this would not be useful. The goal is to release a small interval, generally measured by the \textit{margin of error (\moe)}, equal to half the interval's width.  (I.e., the margin of error for an interval $[a,b]$ is $(b-a)/2$.)

It is acceptable for an algorithm to have coverage \textit{greater} than $1-\alpha$, but generally when this is the case there is some slack in the algorithm, and the interval can be shrunk to obtain lower average margin of error.

The most well known type of confidence interval, and the kind we focus on in this paper, is a confidence interval for the mean of a normal random variable.  In the public setting, this is done with the following algorithm:

\begin{definition}[Confidence Intervals for Normally Distributed Data]
In the case where the sample comes from a normal distribution, $X_1, ..., X_n\sim \mathcal{N}(\mu, \sigma^2)$, the optimal (smallest \moe) ($1 - \alpha$)-coverage confidence interval is
$$c(X)= \left[\bar{X}\pm \frac{s}{\sqrt{n}}q_{n - 1}\left(\frac{\alpha}{2}\right)\right],$$
where $q_{n-1}(\alpha/2)$ is the $\alpha/2$ quantile of the $t$-distribution with $n-1$ degrees of freedom, $s$ is the sample standard deviation, calculated using
$$s^2=\frac{1}{n - 1}\sum_{i=1}^{n}(X_i - \bar{X})^2,$$
and $\bar{X}$ is the sample mean,
$$\bar{X} = \frac{1}{n}\sum_{i = 1}^{n} X_i.$$
\end{definition}

\subsection{Private Confidence Intervals}

Despite the prevalence of confidence intervals in statistical applications, only a handful of papers have attempted to give analysts a way to construct intervals in the private setting.  Awan and Slavković \cite{Awan2018} derive optimal confidence intervals for binomial proportions, and Sheffet \cite{Sheffet2017} describes confidence intervals for private regression coefficients under certain assumptions.  

We are aware of three works that give algorithms for the mean of normally distributed data, those of Karwa and Vadhan \cite{Vadhan2017}, D'Orazio et al.~\cite{DOrazio2015}, and Brawner and Honaker \cite{Honaker2018}.  Because we compare our work to these, we discuss them in more technical detail in Section \ref{sec:ext}.  Additionally, Gaboardi et al.~\cite{gaboardi2018locally} give a method for calculating confidence intervals in the more restrictive \textit{local} differential privacy setting, in part using the same methods as Karwa and Vadhan.\footnote{Gaboardi et al.~give asymptotic rather than exact analysis of their algorithms, much like Karwa and Vadhan.  For the sake of comparison, we implement the work of Karwa and Vadhan and give concrete comparisons.  We do not do this for the Gaboardi et al.~work, as it is similar in design and in a more restrictive setting, presumably therefore achieving worse performance.}

This work on private confidence intervals is directly motivated by attempts to move differential privacy into practice, most specifically the PSI \cite{gaboardi2016psi} project, which attempts to provide an interface for basic statistics about private data sets used in academic research.

The term ``confidence interval'' can be used in the private setting in two different ways, which can be confusing.  It probably helps to consider the following three kinds of confidence intervals:

\paragraph*{Public interval for population mean}  This is the standard sort of confidence interval thoroughly established in statistics.  The goal of the interval is to use the sample data to give an interval estimating the population mean, accounting for the variability induced when selecting the random sample.

\paragraph*{Private interval for sample mean}  This is a common tool in the practice of differential privacy, though rarely discussed in the academic literature.  When a mean is reported (e.g., by adding Laplace noise) it is often helpful to give the analyst an understanding of the uncertainty, so a confidence interval can be constructed.  This interval is meant to show the uncertainty added by the Laplace noise, so it takes into account only the randomness of the private query.  For simple things like the Laplace mechanism, the confidence interval is trivial to construct, though for others it can be very difficult.

\paragraph*{Private interval for population mean}  This is the subject of this paper.  Here the goal is to give an interval that will contain the population mean, but to make that algorithm private.  That means it must account for the noise of both the random sampling and the private query algorithm.

\medskip
It is frequently noted in the differential privacy literature that the noise of private mechanisms is known and can therefore be accounted for in statistical analysis, but this accounting has rarely been studied and is actually quite difficult.  We believe it is unreasonable to expect people who are not privacy experts to do this accounting and that such an expectation deters the use of differential privacy.  Furthermore, some private estimates allow for more accurate noise-accounting than others.  For example, one could estimate mean and standard deviation privately and then use the public confidence interval formula, but such a method can fail to accurately guarantee coverage.  Our goal here is to evaluate private algorithms based on the utility of the final, usable output that a practitioner will want to see.

Private confidence interval algorithms vary not just in what margin of error they produce, but also in what assumptions are required.  These algorithms, both in prior work and in our work, generally require that the data is known to all come from a given range, i.e., $x_i \in [x_{min}, x_{max}]$.  (The algorithms work on more general data, but values below $x_{min}$ are set to $x_{min}$ before other calculations are performed, and similarly for $x_{max}$.)  But some algorithms are very sensitive to that range, degrading in accuracy very quickly if the range is overly wide, while others are insensitive to the range, allowing the analyst to give very conservative values.

Finally, we note that both here and in prior work, the assumption that data is normally distributed is only required for the coverage guarantees of the algorithm.  Privacy is required to hold in general for any input, regardless of its distribution.

\paragraph*{Relationship to hypothesis testing}
In the public setting, confidence intervals are often discussed interchangeably with hypothesis testing, specifically a t-test.  A t-test asks whether a population with a hypothesized mean could plausibly give rise to a sample with the observed sample mean.  The probability that such a mean (or a more extreme difference) emerges is a $p$-value.  It is typical to check whether $p < 0.05$, and a $p$-value less than 0.05 will occur precisely when the hypothesized mean is outside the resulting confidence interval.  As such, any algorithm that produces a confidence interval also produces a hypothesis test and vice versa.

In the private setting, this equivalence no longer holds.  Given a hypothesis test, one normally converts to a confidence interval by asking, ``At what value does the hypothesis test start rejecting?''  But private hypothesis testing algorithms give a $p$-value at one point and cannot be run repeatedly without losing privacy, so there is no way to find where the cutoff for rejecting would be.  As a result, the work on private hypothesis testing in this setting \cite{couch2019differentially, solea2014differentially} does not help us here.

\section{Prior Work}\label{sec:ext}

Here we describe the three existing works that seek to provide private confidence intervals for the population mean of normally distributed data.  We give a very abridged overview of each algorithm; interested readers should refer to the original works for more detail.

\subsection{Karwa and Vadhan}

The most mathematically sophisticated work comes from Karwa and Vadhan \cite{Vadhan2017}. They give algorithms for both the $\delta=0$ and the $\delta>0$ case.  In both cases, their algorithm begins by running a private histogram algorithm on the data and uses that to estimate a range for the data.  The data is then clamped to that range.  Given that truncation, Laplace noise can be added to give a private estimate of the mean and a private but very conservative estimate of variance, which are then used to construct a confidence interval.

This work is impressive and has some very useful results.  The margin of error is shown to be asymptotically optimal, and the coverage guarantees hold for finite $n$, rather than asymptotically\footnote{This is somewhat misleading. For low $n$ the algorithm will output $\bot$ instead of a confidence interval.  Coverage is correct whenever there is output, but the output is withheld for low $n$ because coverage would not be correct in those cases.}.  For the $\delta>0$ case, they require no a priori bounds on the data.  In the $\delta=0$ case, they do require such bounds\footnote{In fact, they need bounds on both the mean and standard deviation, which is stronger than simply having bounds that hold with high probability on the minimum and maximum data values.}, but the accuracy is not highly sensitive to the bounds, so they can be set very conservatively.

However, this work also has serious limitations.  While its asymptotic performance is excellent, its practical performance is unacceptable.  To quote the paper, ``our algorithms are not optimized for practical performance, but rather for asymptotic analysis of the confidence interval length. Initial experiments indicate that alternative approaches (not just tuning of parameters) may be needed to [release] reasonably sized confidence intervals.''  These alternative approaches are exactly what we seek to deliver in this paper.

\begin{algorithm}
\caption{Vadhan}
\begin{algorithmic}[1]
\REQUIRE $\db, \alpha_0, \alpha_1, \alpha_2, \alpha_3, \eps_1, \eps_2, \eps_3,\linebreak \bar{s}_{\text{min}}, \bar{s}_{\text{max}},\bar{X}_{\text{min}},\bar{X}_{\text{max}}$
\STATE $\tilde{X}_{\text{min}},\tilde{X}_{\text{max}}\gets \algn{RANGEFINDER}(\db, \alpha_3, \eps_3, \bar{s}_{\text{min}}, \bar{s}_{\text{max}},\linebreak\bar{X}_{\text{min}},\bar{X}_{\text{max}})$
\STATE Clamp \db by $\tilde{X}_{\text{min}},\tilde{X}_{\text{max}}$
\STATE $\tilde{X}_{\text{var}}\gets\frac{\tilde{X}_{\text{max}} - \tilde{X}_{\text{min}}}{\eps_1 n}$
\STATE $\tilde{s}_{\text{var}}\gets\frac{(\tilde{X}_{\text{max}} - \tilde{X}_{\text{min}})^2}{\eps_2 (n-1)}$
\STATE $\tilde{X}\gets \bar{X} + L_1$, where $L_1\sim\lap(\tilde{X}_{\text{var}})$
\STATE $\tilde{s}^2\gets s^2 + \tilde{s}_{\text{var}}\ln(\frac{1}{2\alpha_2}) + L_2$, where $L_2\sim\lap(\tilde{s}_{\text{var}})$
\STATE $\moe\gets\sqrt{\frac{\tilde{s}^2}{n}}\qt_{n - 1}(1 - \frac{\alpha_0}{2}) + \tilde{X}_{\text{var}}\ln(\frac{1}{\alpha_1})$
\ENSURE $[\tilde{X} - \moe, \tilde{X} + \moe]$
\end{algorithmic}
\end{algorithm}

\subsection{D'Orazio, Honaker, and King}

D'Orazio, Honaker, and King \cite{DOrazio2015} give several private algorithms intended for use in social science research. The only confidence interval they explicitly outline is for the difference of means between two normally distributed variables; however, their method can be easily adapted to produce intervals for a single variable.

Their algorithm first uses a simple Laplace mechanism query to estimate the sample mean.  To get a confidence interval, one needs to compute not just this mean but also an estimate of the sampling variability of the sample mean.  To do this, they use an algorithm similar to one of Smith \cite{Smith2011}.  They first divide the sample into disjoint subsamples and from each calculate such an estimate.  This set of estimates $S$ is then fed to a private quantile algorithm to get estimates of the $25^{\text{th}}$ and $75^{\text{th}}$ quantiles.  This gives an interquartile range estimate $r$ equal to their difference and a center estimate $c$ equal to their average.  The values of $S$ are then truncated to be in $[c-2r, c+2r]$.  The mean of $S$ is then computed with Laplace noise to ensure privacy.

Once there are estimates of the mean and sampling variability, simulated data can be used to compute an actual confidence interval.  We use a similar simulation technique for our algorithms, so we refer the reader to Algorithm \ref{alg:CCI} and surrounding discussion for more detail.

\subsection{Brawner and Honaker}

Given a sample mean, statisticians can estimate the variance of that sample mean using bootstrap resampling \cite{Efron1992Boot}. Given a database $x$, a bootstrap sample $y$ of the same size can be computed by randomly sampling from $x$ with replacement.  The mean of $y$ is then computed and the process is repeated many times.  The distribution of the means of those bootstrapped samples is a good approximation of the sampling distribution of the mean of the original database.  For large $n$, the distribution is known to be asymptotically normal, so the variance of the bootstrapped samples is sufficient to allow the computation of a confidence interval for the population mean.

Despite its prevalance in the practice of statistics, we are familiar with only one (unpublished) paper on private bootstrapping, that of Brawner and Honaker \cite{Honaker2018}.  They give a method that releases $k$ means of bootstrapped samples, each with $1/k^\text{th}$ of the privacy budget.  These are used to calculate a variance estimate, and they're also averaged to create an estimate of the sample mean.  Crucially, it is shown that the sample mean estimate arrived at this way is just as accurate as one computed directly, but this method avoids the need to allocate part of the budget to the sample mean computation.  Given these mean and variance estimates, a confidence interval can be computed.  Unfortunately, the variance estimate can often be too low, resulting in unacceptable coverage, but they give a method to conservatively increase the variance estimate and achieve acceptable coverage.

This result is achieved under \textit{zero-concentrated differential privacy} (zCDP) \cite{Bun}.  This privacy is parameterized by $\rho$, and for any $\delta>0$ one can convert a guarantee of $\rho$-zCDP into a guarantee of $(\eps, \delta)$-differential privacy. 

\section{Algorithms}\label{sec:ours}

We introduce five algorithms to construct \edp{\eps} confidence intervals for the mean of normally distributed data. We start with the simple case, which uses the Laplace Mechanism to produce private estimates of mean and standard deviation. For the second, we modify this method to utilize an alternative dispersion metric that results in a function with lower sensitivity. The remaining three algorithms rely on the exponential mechanism to generative private estimates of quantiles of the data. Each of these methods take a different approach for turning these quantiles into estimates of the center and spread of the sample data, which are then used to construct the confidence interval.

\subsection{Noisy Mean and Variance}

Our first approach is a direct application of Laplacian noise to the sample mean and variance. The noisy mean and variance are then used to construct the appropriate confidence interval. In this algorithm and several future algorithms, $\rho$ is an allocation parameter that determines the fraction of $\eps$ used at various stages of the algorithm. We optimize this value experimentally.

\begin{algorithm}
\caption{Noisy Mean and Variance, \cilap}\label{alg:Lap}
\begin{algorithmic}[1]
\REQUIRE $\db, \eps, \rho$
\STATE $\tilde{\db} \gets \bar{\db} + L_1 \text{ where } L_1 \sim \Lap\left(\frac{\drange}{\rho\eps n} \right)$
\STATE $\tilde{s} \gets \sqrt{\text{max}\left(0,s^2 + L_2\right)},\hfill\linebreak \text{ where } L_2 \sim \Lap \left(\frac{\drange^2}{(1-\rho)\eps n} \right)$
\ENSURE $\tilde{x}, \tilde{s}$
\end{algorithmic}
\end{algorithm}

Here $\db$ is any database, $\eps$ is the privacy parameter, and $0 \leq \rho \leq 1$ is the allocation parameter of $\eps$ among queries. We define the sample mean $\bar{\db} = \frac{1}{n}\sum_{i = 1}^n \db_i$, and the sample variance $s^2 = \frac{1}{n-1}\sum_{i=1}^n (\db_i - \bar{\db})^2$.

\begin{lemma}
$\Delta\bar{\db} = \frac{\drange}{n}$. % leq?
\end{lemma}
\begin{proof}
Given $\db,\db'\in \dbu$ that are neighbors, the entry that was changed between them can only change value by at most $\drange$. As other entries remain unchanged, the sum of all entries can be changed by at most $\drange$. When taking the mean of these databases then, the mean can only change value by at most $\frac{\drange}{n}$. Thus $\Delta\bar{\db} \leq \frac{\drange}{n}$.\\
\end{proof}

\begin{lemma} \label{lm:1}
$\Delta s^2\leq\frac{\drange^2}{n}$.
\end{lemma}
\begin{proof}
This proof modifies Honaker's sensitivity proof for variance estimator $\hat{\sigma}^2 = \frac{1}{n}\sum_{i = 1}^n (x_i - \bar{x})^2$. In its place we use the unbiased sample variance $s^2 = \frac{1}{n-1}\sum_{i=1}^n (x_i - \bar{x})^2$. See Appendix \ref{sec:pfs} for the full proof.
\end{proof}

\begin{theorem}
$\cilap$ is \edp{\eps}.
\end{theorem}

\begin{proof}
By composition, it is sufficient to show that $\tilde{\db}$ is \edp{\rho\eps} and $\tilde{s}$ is \edp{(1-\rho)\eps}. In step 1 of $\cilap$, the amount of noise added is exactly $\Lap(\frac{\Delta\bar{\db}}{\rho\eps})$, thus $\tilde{\db}$ is \edp{\rho\eps} by the properties of the Laplace Mechanism. Similarly, in step 2 of $\cilap$ the amount of noise is $\Lap(\frac{\Delta s^2}{(1-\rho)\eps})$, so $s^2 + L_2$ is \edp{(1-\rho)\eps} and by post-processing, $\tilde{s}$ is \edp{(1-\rho)\eps}.
\end{proof}

\paragraph*{Simulation} When generating public confidence intervals, a $t$-distribution is used to find the critical value and margin of error for a confidence interval. In our case, however, the addition of Laplacian noise to both the standard deviation and the sample mean, renders the $t$-distribution no longer appropriate. Reed in 2006 \cite{NormalLaplace} introduces the useful Normal-Laplace distribution and provides its cdf and pdf. However, as we fail to invert the cdf to construct a quantile function, the distribution cannot be used for constructing a confidence interval. In its place, we use simulation to construct the reference distribution for the noisy sample mean with standard deviation $\tilde{s}$. The margin of error is then estimated to be half of the difference between the $\frac{\alpha}{2}$ quantile and the $1-\frac{\alpha}{2}$ quantile of the simulated reference distribution. Let $q(\db,\alpha)$ be a non-private empirical quantile function that outputs the $\alpha$ quantile of sample $\db$.

\begin{algorithm}
\caption{Confidence interval simulation, \cci}\label{alg:CCI}
\begin{algorithmic}[1]
\REQUIRE $\alpha, \mathcal{A}, \db, \eps$, \textit{nsim}
\STATE $\widetilde{\db}, \widetilde{s} \gets \mathcal{A}(\db, \eps)$
%\STATE $n \gets \textrm{length}(\db)$
\FOR{$i$ from 1 to \textit{nsim}}
  \STATE $\db' \gets\db'_0, ..., \db'_n \sim \mathcal{N}(\widetilde{\db}, \widetilde{s}^2)$
  \STATE $\widetilde{\db}'_i \gets \mathcal{A}(\db', \eps)$
\ENDFOR
\STATE $\moe \gets \frac{q(\tilde{\db}', 1 - \frac{\alpha}{2})-q(\widetilde{\db}', \frac{\alpha}{2})}{2}$
\ENSURE $\widetilde{\db} - MoE, \widetilde{\db} + MoE$
\end{algorithmic}
\end{algorithm}

The algorithm outputs a $1-\alpha$ confidence interval through simulation. The input $\mathcal{A}$ can be any algorithm, such as \cilap, that outputs a private estimate of mean and standard deviation when given a database \db.

Since \cci only interacts with database $\db$ at step 1 through algorithm $\mathcal{A}$, \cci is a post-processing algorithm and preserves $\eps$-differential privacy. In the following sections we focus only on algorithms which produce private estimates of mean and standard deviation. This allows a general application of \cci to construct confidence intervals.

\subsection{Noisy Absolute Deviations}
While a private estimate of the standard deviation can be made by adding noise to the naive estimator, previous work has shown that one can increase utility by using an alternative estimator, one with lower sensitivity\cite{new_anova}.

\begin{definition}[Mean Absolute Deviation]
The mean absolute deviation of the sample $\db_1, \ldots, \db_n$ is
$$\frac{1}{n}\sum_{i=1}^n |\db_i - \bar{\db}|.$$
For normally distributed data $\db_1, \ldots, \db_n$, the ratio of mean absolute deviation to standard deviation is $\sqrt{\frac{2}{\pi}}$.
\end{definition}

The mean absolute deviation has lower sensitivity than $s^2$, therefore reducing the amount of noise necessary to maintain privacy. We use $\frac{1}{n}\sum_{i=1}^n |\db_i - \bar{\db}| \cdot \sqrt{\frac{\pi}{2}}$ to then convert the mean absolute deviation to the sample standard deviation.

\begin{algorithm}
\caption{Noisy absolute deviations, \abslap}\label{alg:AbsLap}
\begin{algorithmic}[1]
\REQUIRE $\db, \eps, \rho$
\STATE $\widetilde{\db} \gets \Bar{\db} + L_1$, where $L_1 \sim \Lap\left( \frac{\drange}{\eps_{1}n} \right)$
\STATE $\widetilde{s} \gets\sqrt{\frac{\pi}{2}}\cdot\text{max}(0, \frac{1}{n}\sum_{i=1}^{n}|\db_i-\bar{\db}|  + L_2)$, where $L_2 \sim \Lap\left(\frac{2\drange}{(1-\rho)\eps n}\right)$
\ENSURE $\widetilde{\db}, \widetilde{s}$
\end{algorithmic}
\end{algorithm}

The sensitivity of sample standard deviation is $\sqrt{\frac{\drange^2}{n}}$ while the sensitivity of the mean absolute deviation is $\frac{2\drange}{n}$ which is asymptotically smaller.

\begin{lemma}
$\Delta\widetilde{s}\leq\frac{2\drange}{n}$.
\end{lemma}
\begin{proof}
We show in Appendix \ref{sec:pfs} that the sensitivity of $\sum_{i = 1}^n |\db_i - \bar{\db}|$ is bounded by $2\drange$. Thus the sensitivity of $\tilde{s}_0$ is bounded by $\frac{2\drange}{n}$, and by post-processing, $\Delta\widetilde{s}\leq\frac{2\drange}{n}$.
\end{proof}

\begin{theorem}
$\abslap$ is \edp{\eps}.
\end{theorem}

\begin{proof}
It is sufficient to show that $\widetilde{s}$ is \edp{(1-\rho)\eps}, since from the above section we already know that $\widetilde{\db}$ is \edp{\rho\eps}. In step 2 of $\abslap$ the noise added is $\Lap(\frac{\Delta\widetilde{s}}{(1-\rho)\eps})$. Thus $\widetilde{s}$ follows the Laplace Mechanism and is \edp{(1-\rho)\eps}.
\end{proof}

Using the estimates of mean and standard deviation given by algorithm $\abslap$, we can construct a private confidence interval using algorithm $\cci$.

\subsection{Exponential Quantiles}
The construction of confidence intervals for the mean of normal data requires some measure of the center and spread of the data. The first two methods address this fairly directly, using noisy estimates of the two statistics. A more flexible approach to estimating center and spread is to work with quantiles. A single quantile is a location statistic, telling us the magnitude of a particular part of the distribution (e.g. the 10\textsuperscript{th} quantile tells us where the left tail is; the 50\textsuperscript{th} tells us where the center is). With two quantiles, it becomes possible to estimate the spread of the data.

We introduce three approaches to computing the mean and standard deviation of data using private quantiles. All three methods rely on the algorithm below, which outputs \edp{\eps} estimates of any desired sample quantiles using the exponential mechanism.  (This algorithm first appears in print in the work of Smith \cite{Smith2011}, who credits it to McSherry and Talwar \cite{mcsherry2007mechanism} and personal correspondence.  We know of no published proof of its privacy.  Under the unbounded differential privacy definition the proof is almost trivial, but the proof in the bounded case is a bit more complex, so we include it here.)

To concisely explain the algorithm, we rely on the following notation:

\begin{itemize}
\item Allow that the values of any given database, $\db$, are sorted: $\db_1 \leq \ldots \leq \db_n$.
\item Allow for notational convenience that $\db_0 = \xmin$ and $\db_{n + 1} = \xmax$.
\item Define bins $B_0,\ldots,B_n \subseteq [\xmin, \xmax) \text{ s.t. } B_i = [\db_i, \db_{i + 1})$.
\item Let $m$ indicate the rank of the quantile of interest. (I.e., $\db_m$ is the ideal output.)
\end{itemize}

Since this algorithm uses the exponential mechanism, we must define its utility function. This function is selected such that utility increases as the number of datapoints that lie between a potential response and the true quantile decrease. This results in values directly adjacent to the true quantile, $\db_m$, having the highest utility.

\noindent Let utility function $U_m:\dbu\times[\xmin,\xmax)\rightarrow\mathbb{R}$ s.t. $\forall \, i\in\{0,\ldots,n\}$ and for all possible responses $y\in B_i$,
$$U_m(\db, y) = U_m(\db, B_i) = \begin{cases}
  i + 1 - m &\text{if } i < m\\
  m - i &\text{if } i \geq m\\
\end{cases}.$$

\begin{algorithm}[H]
\caption{Exponential quantile, \expmed}\label{alg:ExpMed}
\begin{algorithmic}[1]
\REQUIRE $\db, m, \eps$
\STATE $\textbf{Define } B_0,\ldots,B_n \text{ as above}$
%\STATE $\textbf{Define } p_0,\ldots,p_n: $
\FOR{$i$ from 1 to $n$}
\STATE $p_i \gets |B_i|\cdot\text{exp}(\frac{\eps}{2}\, U_m(\db, B_i))$
\ENDFOR
\STATE $\text{Normalize } p_0, \ldots, p_n$ s.t. $\sum_{i = 0}^n p_i = 1$
\STATE Sample $i\in [0, n]$ from the distribution defined by $p_0,\ldots,p_n$
\ENSURE $Y\sim \text{Unif}(B_i)$
\end{algorithmic}
\end{algorithm}

The range $[\xmin, \xmax)$ is split into $n+1$ bins $[\db_i, \db_{i+1})$ where $i \in \{0,\ldots,n\}$. Each bin is assigned a utility score based on its distance from the quantile of interest. Then the exponential mechanism is used to select a bin. The algorithm then randomly outputs a number from the range of the selected bin.

We now show that \expmed is \edp{\eps}.

\begin{lemma}
Given $\db, \db'\in\dbu$ that are neighbors, and their intervals $B_0,\ldots,B_n,\ B'_0,\ldots,B'_n$, if $y\in B_i$ then $y\in B'_{i - 1}\cup B'_i \cup B'_{i + 1}$.
\end{lemma}
\begin{proof}
Let $\db^*$ be a database that has one entry less than both $\db$ and $\db'$, and let $\db_j$, $\db'_k$ be the removed entries of $\db$, $\db'$ such that
\begin{align*}
&\db_0,\ldots,\db_{j - 1},\db_{j + 1},\ldots,\db_{n + 1} = \db^*_0,\ldots,\db^*_n\\
&\indent=\db'_0,\ldots,\db'_{k - 1},\db'_{k + 1},\ldots,\db'_{n + 1}.
\end{align*}
So $B^*_0,\ldots,B^*_{n - 1}\\
\indent= B_0,\ldots,B_{j - 2},B_{j - 1} \cup B_j, B_{j + 1},\ldots,B_n$\\
and $B^*_0,\ldots,B^*_{n - 1}\\
\indent= B'_0,\ldots,B'_{k - 2},B'_{k - 1} \cup B'_k, B'_{k + 1},\ldots,B'_n$.
{%\setlength{\mathindent}{0pt}\setlength{\abovedisplayskip}{-6pt}\setlength{\belowdisplayskip}{-6pt}
\begin{flalign*}
\text{So }\forall \, i,\ B_i &\subseteq B^*_{i - 1} \cup B^*_i&\\
\text{and } B^*_i &\subseteq B'_i \cup B'_{i + 1}\\
\implies B_i&\subseteq B'_{i - 1} \cup B'_{i} \cup B'_{i + 1}.
\end{flalign*}
}
Thus if $y\in B_i$, then $y \in B'_{i - 1} \cup B'_{i} \cup B'_{i + 1}$.
\end{proof}

\begin{lemma}
$\Delta U_m = 1$.
\end{lemma}
\begin{proof}
Given $\db,\db'\in\dbu$ that are neighbors, and $y\in B_i$, allow that $\db,\db'$ have quantiles of interest $\db_m, \db'_m$. By the previous lemma, $y \in B'_{i - 1} \cup B'_{i} \cup B'_{i + 1}$. If $i \geq m$, then
{%\setlength{\abovedisplayskip}{-8pt}
\begin{flalign*}
    |U_m(\db, r) - U_m(&\db', r)|=|(m - i) - U_m(\db', r)|&\\
    &\leq|(m - i) - U_m(\db', B'_{i\pm 1})|\\
    &\leq|(m - i) - (m - i \pm 1)|\\
    &=|\pm 1|\\
    &=1.
\end{flalign*}
}
If $i < m$, then
{%\setlength{\abovedisplayskip}{-8pt}
\begin{flalign*}
|U_m(\db, r) - U_m(&\db', r)|=|(i + 1 - m) - U_m(\db', r)|&\\
&\leq|(i + 1 - m) - U_m(\db', B'_{i\pm 1})|\\
&\leq|(i + 1 - m) - (i \pm 1 + 1 - m)|\\
&=|\pm 1|\\
&=1.
\end{flalign*}
}
Thus $\Delta U_m=1$.
\end{proof}

\begin{theorem}
\expmed is \edp{\eps}.
\end{theorem}
\begin{proof}
It is sufficient to show that \expmed follows the exponential mechanism with an output range of $[\xmin, \xmax)$. Let $i,\, p_0,\ldots,p_n$ be the final values of those same variables in \expmed.
{%\setlength{\abovedisplayskip}{-8pt}
\begin{flalign*}
    \Pr[\expmed(\db, m, \eps) = y] &= p_i\, \Pr[Y=y],\ Y\sim\text{Unif}(B_i)\\
    &\propto p_i\frac{1}{|B_i|}\\
    &\propto \text{exp}\left(\frac{\eps}{2}\, U_m(\db,B_i)\right)\\
    &=\text{exp}\left(\frac{\eps\, U_m(\db,y)}{2\Delta U_m}\right).
\end{flalign*}
}
Thus \expmed is $\eps$-differentially private following the exponential mechanism.\\
\end{proof}

\begin{theorem}\label{thm:unbias}
\expmed outputs an unbiased estimator of the median of symmetric data.
\end{theorem}
\begin{proof}
We prove Theorem \ref{thm:unbias} by showing that the expected output of \expmed falls into bin $B_m$ or bin $B_{m - 1}$. These bins are adjacent to the median, which implies the expected output will be the median. See the full proof in Appendix \ref{sec:pfs}.

%\ag{Can't this just be done as a two-line proof?  The probability distribution of the data is symmetric, and the algorithm itself is symmetric, so...}

\end{proof}

Since our data are symmetrically distributed, the sample median is an unbiased estimate of the mean. Therefore, the private median produced by \expmed is an unbiased estimate of the mean. The algorithm is slightly biased when estimating other quantiles, but the bias is practically insignificant for moderately sized $n$ or $\eps$. We discuss this bias in Appendix \ref{sec:expq-bias}.

\subsection{Centered Quantiles}
The most straightforward application of \expmed for constructing confidence intervals is to use its median estimate as our mean estimate and use some other quantile as an estimate of standard deviation. Let $\qnorm$ be quantile function of the standard normal distribution and $b$ be the choice of quantile. (We will choose $b$ in practice through experimental optimization.)

\begin{algorithm}
\caption{Centered quantiles, \cexp}\label{alg:ExpM}
\begin{algorithmic}[1]
\REQUIRE $ \db, \eps, \rho, b$
\STATE $\widetilde{\db} \gets \expmed\left(\db, \lfloor\frac{n+1}{2}\rfloor, \rho\eps\right)$
\STATE $d \gets \expmed(\db, \lfloor b(n-1)+1\rfloor, (1-\rho)\eps)$
\STATE $\widetilde{s} \gets \text{max}\left(0, \frac{d - \widetilde{\db}}{\qnorm(b)}\right)$
\ENSURE $\widetilde{\db}, \tilde{s}$
\end{algorithmic}
\end{algorithm}

\begin{theorem}\label{thm:cenq-is-dp}
$\cexp$ is \edp{\eps}.
\end{theorem}
\begin{proof}
$\cexp$ interacts with the database only through two queries to \expmed. The first query uses privacy parameter $\rho\eps$ and the second uses $(1-\rho)\eps$. Thus by composition and post-processing, $\cexp$ is \edp{\eps}.
\end{proof}

\subsection{Symmetric Quantiles}

Here we take a different approach measuring the center of the data. We use algorithm \expmed to compute two different quantiles an equal distance away from the median. The average of these quantiles is used to estimate the mean while the difference is used to estimate standard deviation.

\begin{algorithm}
\caption{Symmetric quantiles, \cdbl}\label{alg:DblM}
\begin{algorithmic}[1]
\REQUIRE $\db, \eps, b$
\STATE $d_1 \gets \expmed\left(\db, \lfloor b(n-1)+1 \rfloor, \frac{\eps}{2}\right)$
\STATE $d_2 \gets \expmed\left(\db, \lfloor (1-b)(n-1) + 1) \rfloor, \frac{\eps}{2}\right)$
\STATE $\widetilde{\db} \gets \frac{d_1+d_2}{2}$
\STATE $\widetilde{s} \gets \text{max}\left( 0, \frac{d_2 - \widetilde{\db}}{\qnorm(1-b)} \right)$
\ENSURE $\widetilde{\db}, \widetilde{s}$
\end{algorithmic}
\end{algorithm}

\begin{theorem}
\cdbl is \edp{\eps}.
\end{theorem}

The proof proceeds in an analogous manner to that of Thm.~\ref{thm:cenq-is-dp}.

\subsection{Median of deviations}

The final approach uses algorithm \expmed to first compute the median as our estimate of the mean. To estimate standard deviation, we compute the distance between every datapoint and the estimated mean, then we take the median of these distances to estimate standard deviation.

\begin{algorithm}
\caption{Median of deviations, \mad}\label{alg:MAD}
\begin{algorithmic}[1]
\REQUIRE $\db, \eps, \rho$
\STATE $\widetilde{\db}\gets \expmed(\db,\lfloor\frac{n+1}{2}\rfloor,\rho\eps)$
\STATE $\db'\gets|\db_1-\widetilde{\db}|,\ldots,|\db_n-\widetilde{\db}|$
\STATE $\widetilde{s}\gets\frac{\expmed(\db',\lfloor\frac{n+1}{2}\rfloor, (1-\rho)\eps)}{\qnorm(.75)}$
\ENSURE $\widetilde{\db}, \widetilde{s}$
\end{algorithmic}
\end{algorithm}

\begin{theorem}
\mad is \edp{\eps}.
\end{theorem}
\begin{proof}
\mad is another algorithm that composes and post-processes two queries to \expmed. The interaction in step 2 is private since the database is only being modified element-wise by $\widetilde{\db}$, and the only information read from this new database is through the private query to \expmed at step 3. Thus \mad is \edp{\eps}.
\end{proof}

\section{Experimental Results}\label{sec:rslts}
Because we are focused on concrete performance at low $n$, rather than asymptotic analysis, we must evaluate our algorithms experimentally.  The first thing we must do is experimentally optimize the parameters of each algorithm.  Having done that, we must check that they output confidence intervals with the promised coverage.  Finally, we must compare them to find the best algorithm(s) and then compare those to prior work.

\subsection{Parameter Optimization}
The algorithms \cilap, \abslap, \cexp, and \mad all require an $\eps$ allocation parameter than determines what proportion of the privacy budget is consumed at different steps of the algorithm. Additionally, \cexp and \cdbl have a parameter $b$ which corresponds to the quantile(s) used to estimate the standard deviation of the database. In all cases, optimization was done experimentally by varying $\rho$ or $b$.  See Appendix \ref{sec:alloc} for figures demonstrating the experimental results.  In principle, the optimal parameter could be different for different choices of $n$, $\eps$, or range, but we found that in all cases we could pick $\rho$ or $b$ values that were roughly optimal in all cases.  (In many cases, there was a large region of choices that seemed roughly equally good.)  These parameters were fixed at the values given below, and all the following results use these parameter values.

\smallskip
\begin{center}
\begin{tabular}{|c|c|}
%\hline 
Algorithm & Parameters \\ 
\hline 
$\cilap$ & $\rho=0.8$ \\ 
\hline 
\abslap & $\rho=0.85$ \\ 
\hline 
\cexp & $\rho=0.5, b=0.65$ \\ 
\hline 
\cdbl &  $b=0.35$\\ 
\hline 
\mad & $\rho=0.5$ \\ 
%\hline 
\end{tabular} 

\end{center}

\subsection{New Algorithms}
Our first experiments sought to determine which of our own algorithms performed best. We compared them with respect to their average \moe while varying other parameters.

Figure \ref{fig:ours_moe_ex}a shows the \moe of our algorithms at $\eps = .01$ and $(\xmin, \xmax) = (-6, 6)$.  By a database size of roughly 1000, \cdbl is the clear winner, while for smaller databases \abslap was best.  (We see that in general \abslap and \cilap are almost identical in \moe, with \abslap consistently having an extremely slight edge.)  Figure \ref{fig:ours_moe} in Appendix \ref{sec:gphs} shows that these findings are generally consistent across choices of $\eps$.\footnote{There are some exceptions with extremely low $\eps$ values and very wide ranges, but our goal here is to find generally useful algorithms, not ones that are marginally less horrible in a weird corner case where everything is bad.}  As a rule of thumb, we find that \cdbl is the superior algorithm once $n > 100/\eps$.

Figure \ref{fig:ours_moe} also shows results as we vary the $[\xmin, \xmax]$ range in which the data is bounded.  Recall that this is a range given by the analyst, and all data outside the range is moved inside it (set equal to \xmin or \xmax) before the algorithm is applied.  Recall also that some of the prior work goes through great pains to ensure that this range is not needed or can be set very conservatively.  The Laplace noise algorithms, \abslap and \cilap, are sensitive to this range, since their noise is proportional to the range.  As the range increases their \moe increases significantly.

Significantly, we find that our quantile-based algorithms are not sensitive to this range as long as $n$ is not very low.  Widening the range only increases the probability that the quantile algorithms pick the most extreme bucket from which to sample.  This bucket increases in width proportionately to the range, but its utility decreases exponentially with $n$.  So for reasonably small $n$, this exponential utility decrease is great enough to make the probability of picking the outermost bucket vanishingly small even when the range is set extremely conservatively.   Conveniently, this effect also seems to show itself by the time $n$ is approximately $100/\eps$.  Therefore we know that in that regime \cdbl is an extremely precise algorithm that requires only extremely minimal knowledge of the analyst.  Figure \ref{fig:ours_moe_ex} shows this for the $\eps=0.1$ case.  Compared to Figure \ref{fig:ours_moe_ex}, the Laplace noise-based algorithms suffer greatly with a wider range, while the quantile-based algorithms are unaffected.  Figure \ref{fig:ours_moe} in Appendix \ref{sec:gphs} shows the same thing for other parameter settings.

\begin{figure}[!htb]
    \centering
    \includegraphics[width=\linewidth]{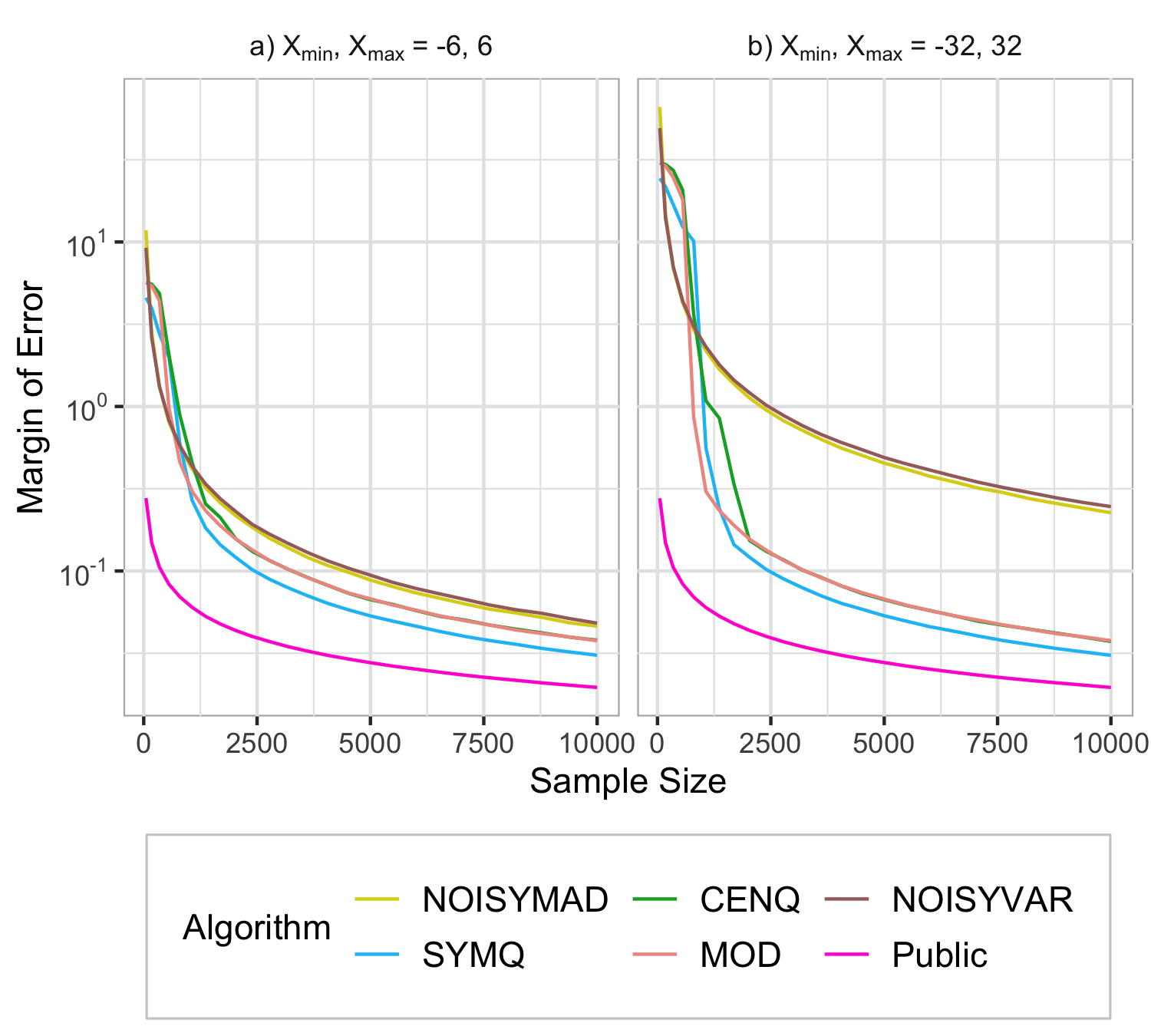}
    \caption{Comparison of our algorithms with respect to their average \moe at various database sizes at two different ranges with $\eps=.1$. The distribution of the underlying database was a standard normal, so an $\xmax$ of 32 corresponds to 32 standard deviations away from the mean.}
    \label{fig:ours_moe_ex}
\end{figure}

We also need to confirm the validity of the algorithms.  That is, we must check that their coverage is truly at least $1-\alpha$.  This must be done because the simulated distribution used to calculate the \moe is based on an estimated standard deviation for the underlying data.  If this estimate is bad enough, it could result in invalid confidence intervals.  To do this we run each test many times at many $\alpha$ values and report the percentage of the time that the true mean was included in the resulting interval.  Figure \ref{fig:ours_cov_range_e_ex} shows one example, and Figures \ref{fig:ours_cov_range_e} and \ref{fig:ours_cov_n_e} in Appendix \ref{sec:gphs} show results at a variety of parameter settings, always with similar (acceptable) results.

What we find is that coverage is generally acceptable.  (In these figures, ``acceptable'' means that the coverage plots never drop below the diagonal.)  The one possible exception is the \cexp (centered quantiles) algorithm.  In some of the experiments it had slightly low coverage for low $\eps$ values.  Because it is so slight, more work be required to determine for sure whether this was a real issue or just experimental noise.  However, \cexp is consistently outperformed by \cdbl anyway, so it doesn't seem worthy of further investigation.

\begin{figure}[!htb]
    \centering
    \includegraphics[width=\linewidth]{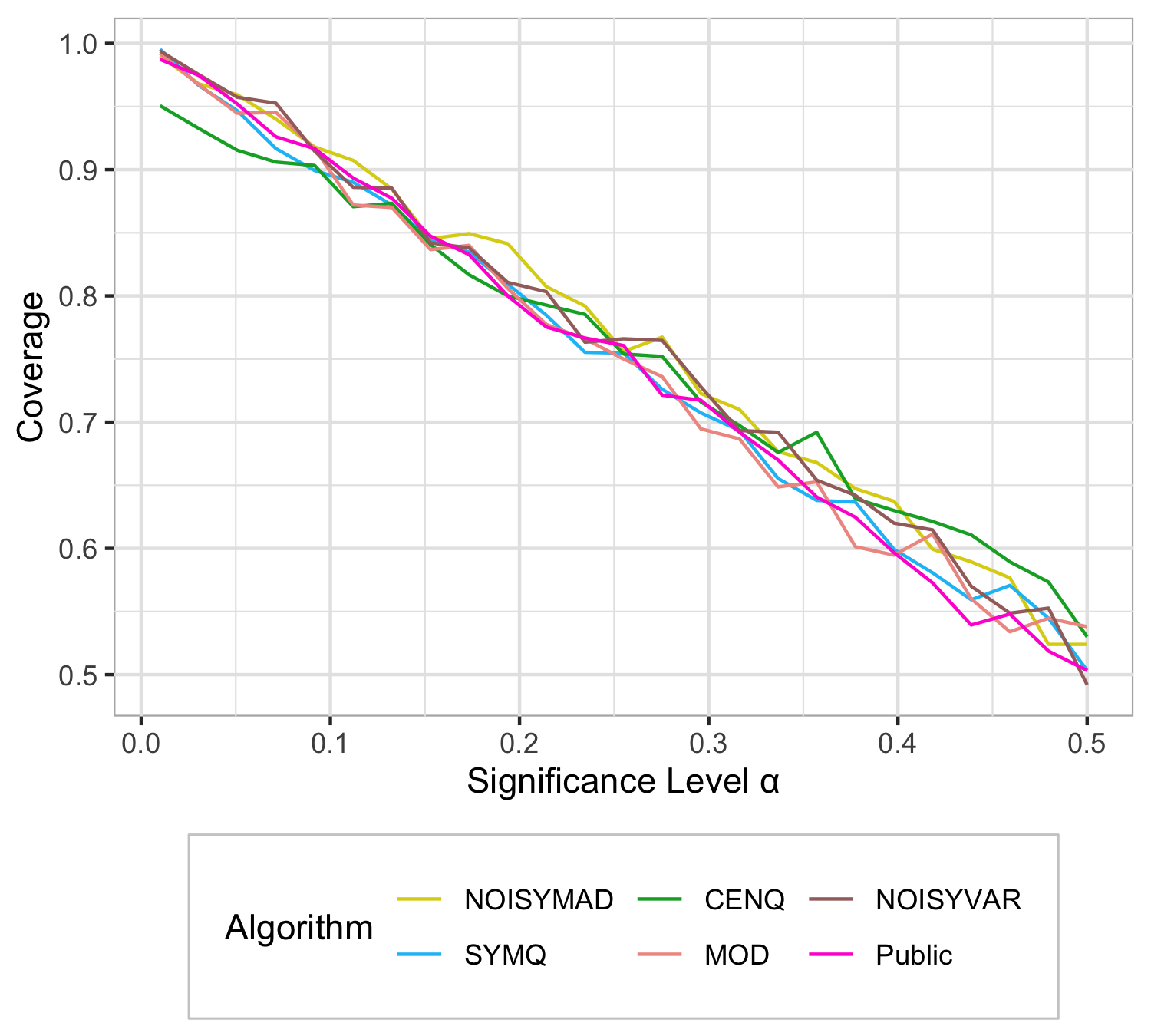}
    \caption{Comparison of our algorithms  with respect to coverage when the range is $[-6,6)$ and \eps$=0.1$. The underlying distribution of the databases was a standard normal distribution. $n=1000$}
    \label{fig:ours_cov_range_e_ex}
\end{figure}

We also check validity in a case where the analyst has not set \xmin and \xmax so well.  In particular, we imagine that the true mean is not centered in the $[\xmin, \xmax]$ range and that potentially one side of the range is close enough to the true mean to clip a significant number of values.  These results can be found in Figure \ref{fig:comp_non0center_cov} in Appendix \ref{sec:gphs}.  We find that \abslap performs fine.  We see no problem with \cdbl when $n>100/\eps$.  For lower values of $n$, \cdbl does display poor coverage, but at those values it is not the superior algorithm anyway.

\subsection{Comparison to Existing Work}
We then compare \cdbl (our best algorithm for $n>100/\eps$) and \abslap (our best algorithm for lower $n$) to the existing work described in Section \ref{sec:ext}.
We used the same experimental framework as we did before, and part of the results we received are compiled in Figure \ref{fig:comp_moe_ex}. The full results are shown in Figure \ref{fig:comp_moe} in Appendix \ref{sec:gphs}.

\begin{figure}[!htb]
    \centering
    \includegraphics[width=\linewidth]{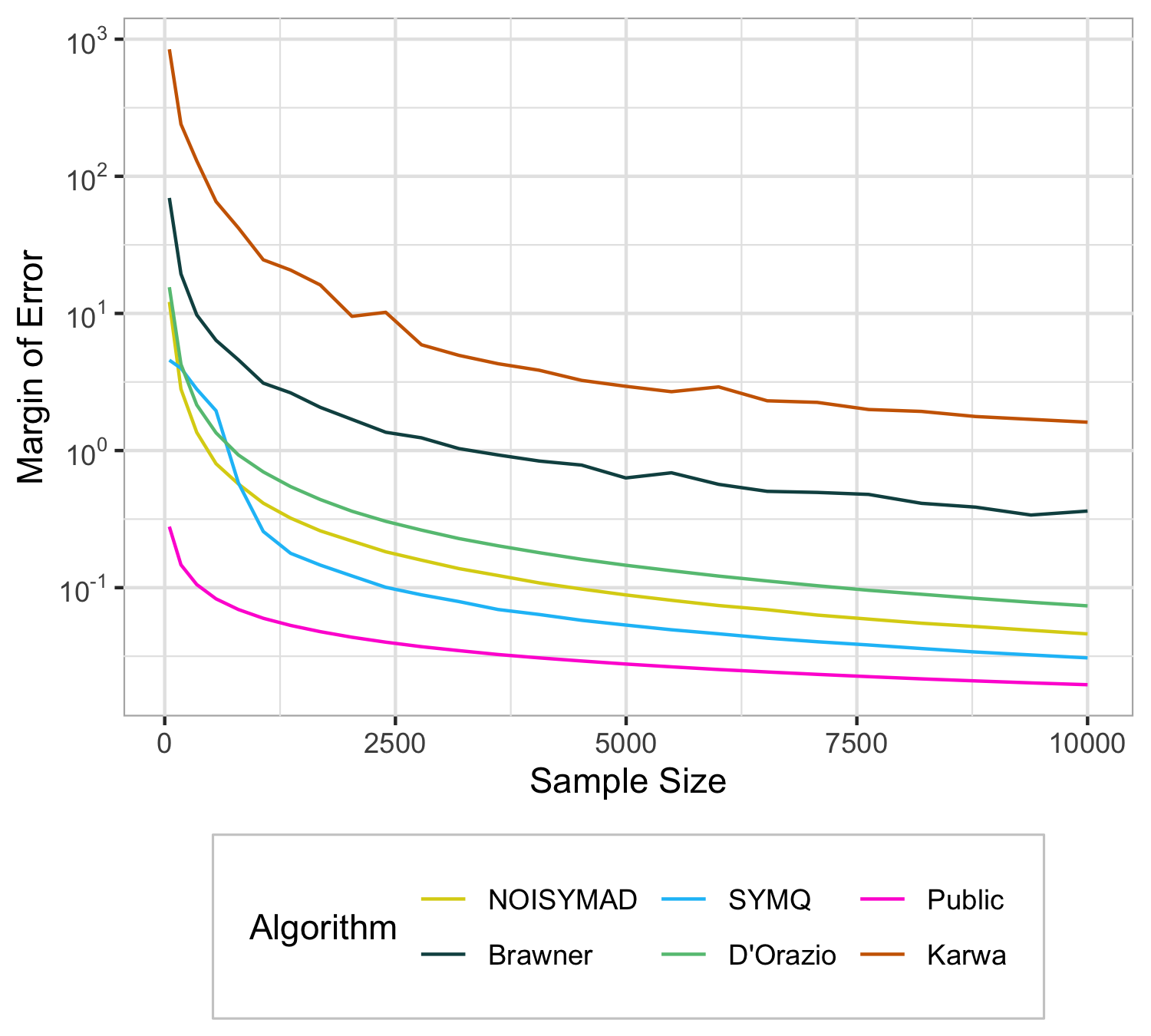}
    \caption{Comparison of our best algorithms to those prior works with respect to their average \moe when the range is $[-6,6)$ and \eps$=0.1$. The underlying distribution of the databases was a standard normal distribution.}
    \label{fig:comp_moe_ex}
\end{figure}

We varied $\eps$, $n$, and the data range, and we found that in all cases the lowest \moe algorithm was one of ours.  In most cases, both \cdbl and \abslap outperformed all prior work. The closest was the work of D'Orazio, Honaker, and King \cite{DOrazio2015}.

We also compared coverage between the various algorithms. We again estimated each algorithm's coverage through simulation, at many different values of $\alpha$. The experimental results of our two best algorithms and the previously existing algorithms are compiled in Figure \ref{fig:comp_cov_ex}. The same as before, our algorithms have coverage of roughly $1-\alpha$, which is ideal.  The Karwa and Vadhan \cite{Vadhan2017} and Brawner and Honaker \cite{Honaker2018} algorithms have extremely broad coverage, being much more conservative than is necessary.  This probably comes from the fact that they use loose upper bounds to set the \moe, rather than precise simulation.
Figure \ref{fig:comp_cov} in Appendix \ref{sec:gphs} contains the same experiment run at a variety of range and $\eps$ values, all with similar results.

\begin{figure}[!htb]
    \centering
    \includegraphics[width=\linewidth]{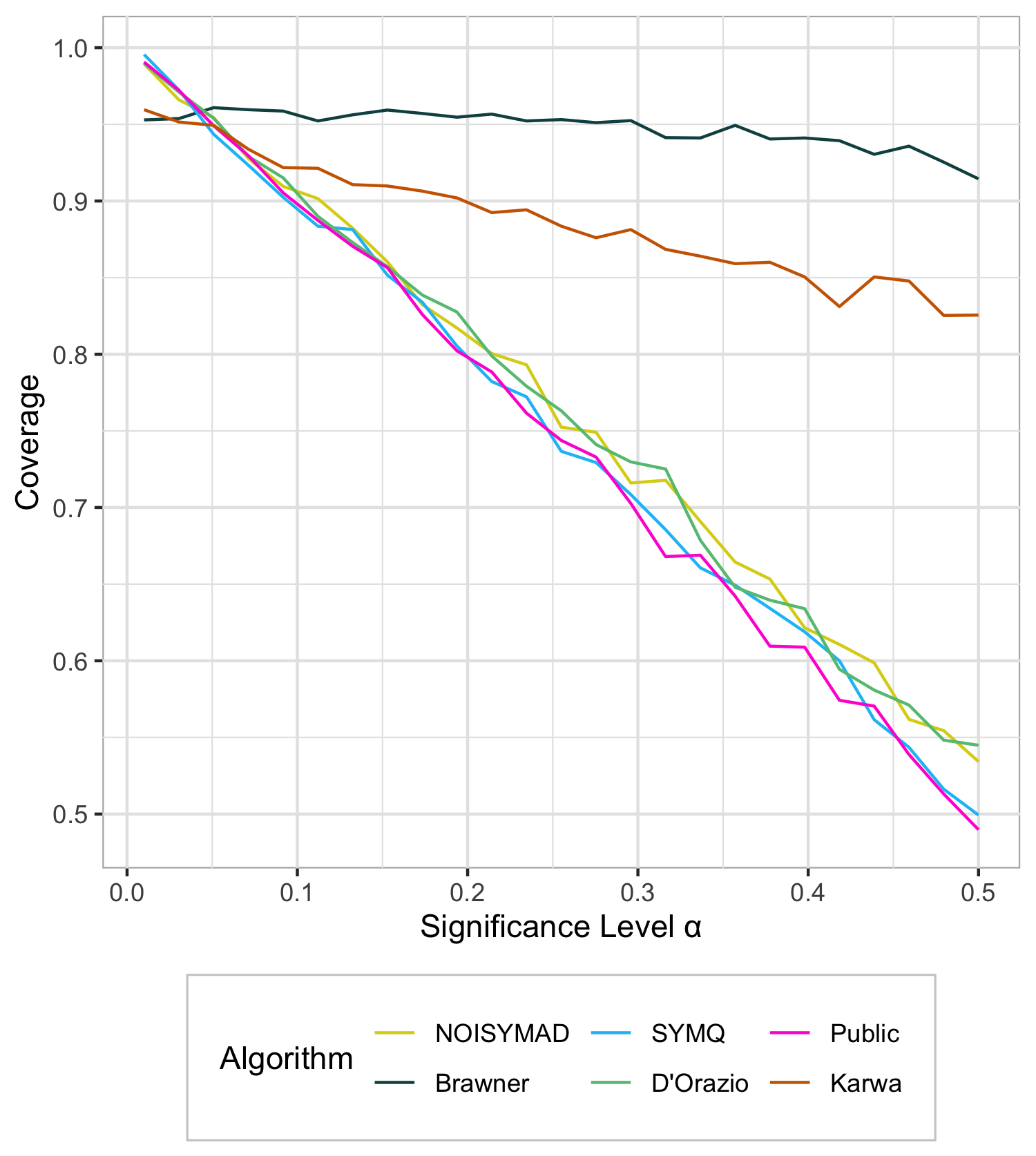}
    \caption{Comparison of our best algorithms to those prior works with respect to coverage when the range is $[-6,6)$ and \eps$=0.1$. The underlying distribution of the databases was a standard normal distribution.}
    \label{fig:comp_cov_ex}
\end{figure}

\section{Discussion}\label{sec:dsc}

We have given two practical algorithms for producing confidence intervals for the population mean of normally distributed data.  As long as $n$ is somewhat large (at least $100\eps$), \cdbl performs very well with little drawback.  It allows the analyst to set the $[\xmin, \xmax]$ window extremely conservatively, and the validity is resilient even to a small mistake on the part of the analyst that clips a portion of the data.  When $n$ is smaller, \abslap is superior (though in this case the analyst must set $[\xmin, \xmax]$ a bit more carefully to avoid adding too much noise).  It is worth taking a moment to think about why the quantile-based method is so useful here.  

Figure \ref{fig:comp_range_cov_01} shows the distribution of center estimates for normally distributed data.  The most accurate (highest-peak) estimate is of course the sample mean.  We also show two private estimates, each at both $\eps=0.1$ and $\eps=0.25$.  One is the standard Laplace mechanism sample mean estimate.  The other is our exponential mechanism-based quantile algorithm, \expmed, used to find the median.  In both cases the quantile algorithm gives a better estimate.

\begin{figure}[!htb]
    \centering
    \includegraphics[width=\linewidth]{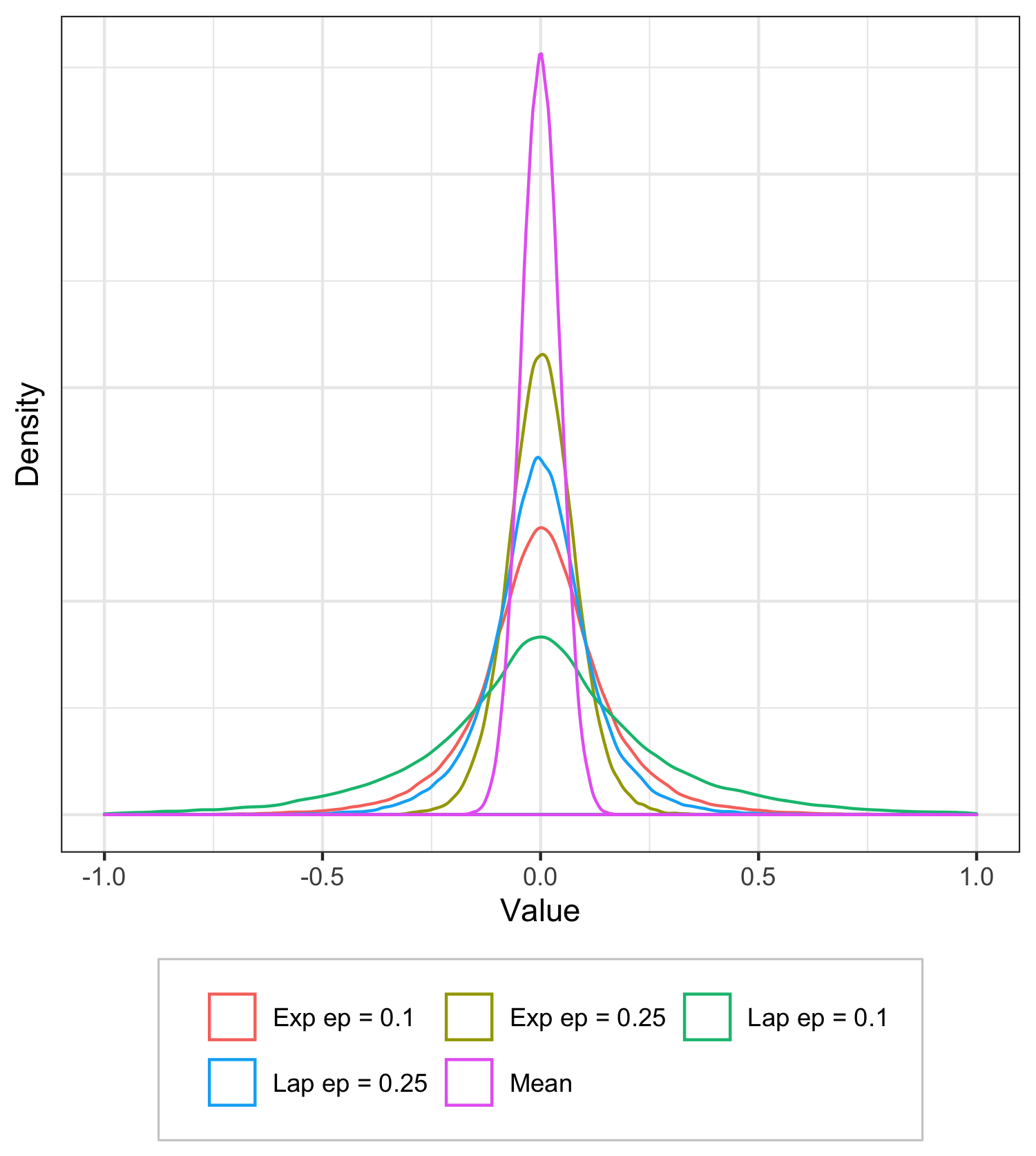}
    \caption{Comparison of the distribution of a sample mean, a private sample mean with  Laplace noise, and our exponential mechanism median estimate. $n = 500$, $\eps = 0.1$ or $0.25$}
    \label{fig:comp_range_cov_01}
\end{figure}

We think this is likely to be part of a larger lesson.  The laplacian mechanism is thought of as the ``best'' algorithm for estimating means, and in some worst-case sense this is true.  Medians are thought to be harder to calculate, since their worst-case sensitivity is high, and more complex algorithms are needed.  But when data is ``nice'' or somewhat predictable, the median algorithm is a better estimate of the mean than the ``best'' mean algorithm.  And much of statistics assume some simple properties of the data distribution anyway, so these assumptions are not additional limitations.  We think private algorithms should often take into account the likely data distribution.  Even when that distribution is not known, using a small portion of the budget for an initial check of possible special cases might often be worthwhile.

In particular, when statistical analysis assumes something about the data, the private version of the analysis should be evaluated under the same assumption.\footnote{As stated before, this is how \textit{utility} should be measured.  Privacy usually should still be a worst-case notion.}  Given this, it makes sense to design special purpose queries that will be more accurate on particular types of data.

In this work, we've given highly practical algorithms for private confidence intervals.  It is noteworthy that our best algorithms are quantile-based, relying on an algorithm that is excellent in this setting but that is \textit{not} an ideal way to measure the center of a set of arbitrary data.

\section{Conclusion}

Our work attempting to find more powerful algorithms for constructing private confidence intervals of the mean of normal data has lead us to several algorithms that perform better at the task than the previously existing work in this area. These algorithms show it is possible in practice to generate small confidence intervals while also providing strong privacy gaurantees. Our best algorithm, the symmetric quantiles algorithm \cdbl, approaches the public confidence interval quite rapidly for moderately sized $n$ and $\eps$. Much of its good performance is due to the exponential quantile algorithm, \expmed, which provides us with more accurate estimates of the mean and standard deviation of a sample than its laplacian noise counterparts. The insensitivity of \expmed to the database range, $\xmin$ and $\xmax$ also allows our confidence interval algorithms that rely on it to give small intervals despite even the most conservative ranges.

\clearpage
\bibliography{sources}

\clearpage
\appendix
\section{Unbiased median estimation}\label{sec:pfs}

\begin{theorem}\label{unbiased}
For a sample drawn from a symmetric distribution with symmetric bounds, $\xmin,\xmax$, \expmed is an unbiased estimator of the median.
\end{theorem}

\noindent Let \db be the expected value of a database drawn from a symmetric distribution with symmetric bounds.\\
Let $\db_1,...,\db_n \in [\xmin, \xmax)$ be database \db sorted where $\db_1 \leq ... \leq \db_n$.\\
Allow for notational convenience that $\db_0 = \xmin$ and $\db_{n + 1} = \xmax$.\\
Let bins $B_0,...,B_n \subseteq [\xmin, \xmax) \sucht B_i = [\db_i, \db_{i + 1})$.\\
Let $\db_m$ be the median of database \db. If $n$ is odd, $m=\frac{n+1}{2}$, if $n$ is even, $m=\frac{n}{2}$.\\
Let utility function $U_m:\dbu\times[\xmin,\xmax)\rightarrow\mathbb{R}$ st $\forall i\in[0,n + 1],\forall r\in B_i$,
$$U_m(\db, r) = U_m(\db, B_i) = \begin{cases}
  i + 1 - m &\text{if } i < m \\
  m - i &\text{if } i \geq m \\
\end{cases}.$$

\begin{lemma}
Given $\delta\in(0,\db_m)\sucht\forall i,\,\db_m\pm\delta\neq \db_i$,
$$\db_m-\delta\in B_{m-i}\implies \db_m+\delta\in B_{m-i+1}.$$
\end{lemma}

\begin{proof}
If $n$ is odd, $\db_m$ is the true median.\\
If $n$ is even, $\db_m$ is not the median, but as $n\rightarrow\infty$, it quickly approaches the true median.\\
Since the distribution of \db is symmetric around the median, $\db_m-\db_{m-i}=\db_{m+i}-\db_m$.\\
Given $\delta\in(0,\db_m)\sucht\forall i,\,\db_m\pm\delta\neq \db_i$,\\
let $i\in[1,m] \sucht \db_m-\delta\in B_{m-i}$.
\begin{flalign*}
\text{So }&\db_{m-i}<\db_m-\delta<\db_{m-i+1}&\\
\implies &\db_m-\db_{m-i+1}<\delta<\db_m-\db_{m-i}\\
\implies &\db_{m+i-1}-\db_m<\delta<\db_{m+i}-\db_m\\
\implies &\db_{m+i-1}<\db_m+\delta<\db_{m+i}\\
\end{flalign*}
So $\db_m+\delta\in B_{m+i-1}$.
\end{proof}

\begin{proof}[Proof of Thm. \ref{unbiased}]
It is sufficient to show that the pdf of the output of $\expmed(\db,m,\eps)$ is symmetric around the median, $\db_m$.\\
Since \expmed follows the exponential mechanism, its pdf must be
$$\phi\cdot\exp\left(\frac{\eps}{2}U_m(\db,r)\right)$$
for some normalization constant $\phi$.\\
Given $\delta\in(0,\db_m)$, suppose without loss of generality that $\forall i,\,\db_m\pm\delta\neq \db_i$.\\
\begin{flalign*}
\text{So }U_m(\db,\db_m-\delta)&=U_m(\db,B_{m-i})\text{ for some }i\in[1,m]&\\
&=(m-i)+1-m\\
&=m-(m+i-1)\\
&=U_m(\db,B_{m+i-1})\\
&=U_m(\db,\db_m+\delta).
\end{flalign*}
\begin{align*}
\text{So }&\phi\cdot\exp\left(\frac{\eps}{2}U_m(\db,\db_m-\delta)\right)\\
&=\phi\cdot\exp\left(\frac{\eps}{2}U_m(\db,\db_m+\delta)\right).
\end{align*}
If $n$ is odd, then the pdf of \expmed is symmetric in the range $\xmin,\xmax$ around the median.\\
If $n$ is even, then the pdf of \expmed is symmetric in the range $\xmin,\db_n$ around the median. The nonsymmetric part of the pdf, $[\db_n,\xmax)=B_n$ has probability density $|B_n|\phi\exp\left(\frac{\eps}{2}U_m(\db,B_n)\right)=|B_n|\phi\exp\left(\frac{\eps}{4}n\right)$. Since this density approaches 0 rapidly as $n\rightarrow\infty$, the pdf of \expmed is asymptotically symmetric around the median.\\

Since the pdf of \expmed is symmetric around the median, \expmed is an unbiased estimator of the median.
\end{proof}

\section{Biased quantile estimation}\label{sec:expq-bias}

As we show above, \expmed is unbiased for estimating the median of normally distributed data. However, this does not hold for other quantiles.\\
Let $X$ be a random database drawn i.i.d. from a normal distribution.\\
Let $q$ be the index of the quantile of interest in the database.\\
So the expected value of $\expmed(X, q, \eps)$, an estimate of the $(\frac{q}{n})^{th}$ quantile, is

$$E[\expmed(X, q, \eps)] = \int_{\xmin}^{\xmax}x \cdot f(x) dx$$

Where $f$ is the probability density function of $\expmed(X, q, \eps)$. Since $f$ is constant within any bin we can greatly simplify this expession as follows, denoting this constant probability as $p_i$ for the $i^{th}$ bin

\begin{align*}
E[\expmed(X, q, \eps)] &= \int_{\xmin}^{\xmax}x \cdot f(x) dx \\
&= \sum_{i = 0}^{n} \int_{X_i}^{X_{i+1}} x \cdot p_i dx \\
&= \sum_{i = 0}^{n} p_i \int_{X_i}^{X_{i+1}} x dx \\
&= \sum_{i = 0}^{n} p_i \frac{X_{i+1}^2 - X_i^2}{2} \\
&= \sum_{i = 0}^{n} p_i (X_{i + 1} - X_i) \frac{X_{i+1} + X_i}{2}
\end{align*}

Notice that $X_{i + 1} - X_i$ is the width of bin $i$ and $\frac{X_{i+1} + X_i}{2}$ is the midpoint of that bin. Although this expression is concise, in practice it is very difficult to work with analylitcally since the distance between $X_i$ and $X_{i+1}$ will depend on complex order statistics. Given index $i$, the expected value of $X_i$ is

$$E[X_i] = \int_{-\infty}^{\infty} (i+1) {{n}\choose{i}} \varphi(x)^{i-1} (1 - \varphi(x))^{n-i} \Phi(x) dx$$

Where $\varphi(x)$ is the pdf of the normal distribution of $X$, and $\Phi(x)$ is the corresponding cdf. It is easy to see how unweildy this expression will make deriving analytical results for our estimator. For this reason, we choose to implement this function and plot the bias in our estimator empirically for various values of $n$ and $\eps$.

\begin{figure*}[h]
    \centering
    \includegraphics[width = 0.65\textwidth]{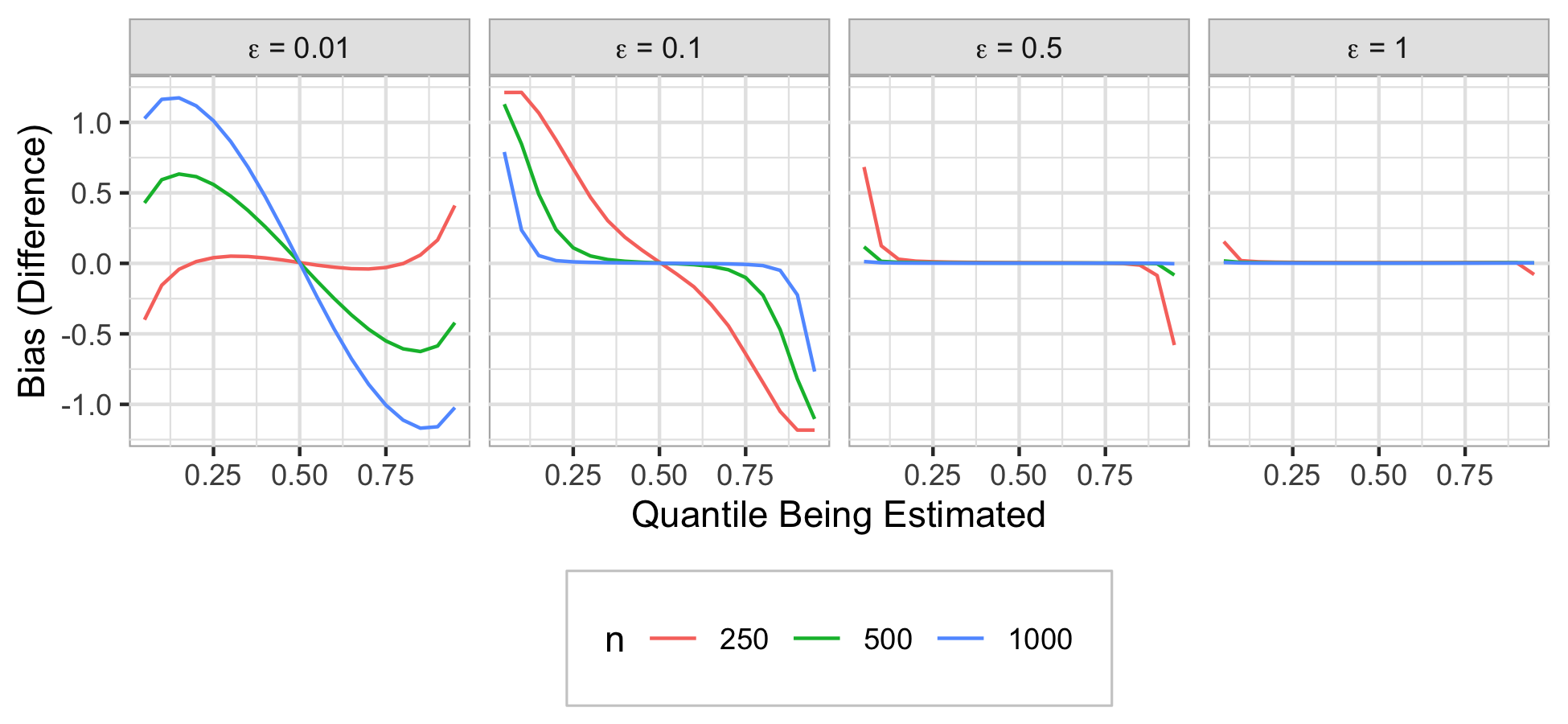} \par
    \caption{The bias (As difference between true and expected value) of our estimator given different values of n and epsilon}
    \label{fig:bias}
\end{figure*}

\begin{theorem}[Variance Sensitivity following Honaker]
$\Delta s^2 = \frac{(\xmax - \xmin)^2}{n}$.
\end{theorem}

\begin{proof}
Let $\db, \db'$ be two neighboring datasets which only differ at the jth row. For simplier notation, we let $\bxmj = \frac{1}{n-1}\sum_{i \neq j} \db_i$. Then the sample variance can be rewritten as

\begin{align*}
    s^2 & = \frac{1}{n-1}\sum_{i=1}^n(\bar{\db} - \db_i)^2\\
    & = \frac{n}{n-1}\bar{\db}^2 - 2\bar{\db}\frac{1}{n-1}\sum_{i=1}^n \db_i + \frac{1}{n-1}\sum_{i=1}^n \db_i^2 \\
    & = \frac{n}{n-1}\bar{\db}^2 - 2\bar{\db}\frac{n}{n-1}\bar{\db} + \frac{1}{n-1}\sum_{i=1}^n \db_i^2 \\
    & = \frac{1}{n-1}\sum_{i=1}^n \db_i^2 - \frac{n}{n-1}\bar{\db}^2 \\
    & = \frac{1}{n-1}\sum_{i=1}^n \db_i^2 - \frac{1}{n(n-1)} \left(\sum_{i=1}^n \db_i\right)^2 \\
    & = \frac{1}{n-1}\left({\db_j}^2 + \sum_{i \neq j} \db_i^2\right) - \\ & \quad \quad  \frac{1}{n(n-1)} \left(\db_j + \sum_{i \neq j} \db_i\right)^2 \\
    & = \frac{1}{n-1}\left({\db_j}^2 + \sum_{i \neq j} \db_i^2\right) - \\ & \quad \quad  \frac{1}{n(n-1)} \left[\db_j^2 + 2\db_j\sum_{i \neq j} \db_i + \left(\sum_{i \neq j} \db_i\right)^2\right] \\
    & = \frac{1}{n-1}\left[\sum_{i \neq j} \db_i^2 - \frac{1}{n} \left(\sum_{i \neq j} \db_i\right)^2\right] + \\ & \quad \quad  \frac{1}{n(n-1)}\left[(n-1)\db_j^2 - 2\db_j\sum_{i \neq j} \db_i\right] \\
    & = \left[\frac{1}{n-1}\sum_{i \neq j} \db_i^2 - \frac{1}{n(n-1)} \left(\sum_{i \neq j} \db_i\right)^2\right] + \\ & \quad \quad  \frac{n-1}{n(n-1)} \left({\db_j}^2 - 2\db_j\bxmj\right) \\
    & = \left[\frac{1}{n-1}\sum_{i \neq j} \db_i^2 - \frac{1}{n(n-1)} \left(\sum_{i \neq j} \db_i\right)^2 \right] + \\ & \quad \quad \frac{1}{n}\left({\db_j}^2 - 2\db_j\bxmj\right)
\end{align*}

Thus changing $\db_j$ would only affect the latter term $\frac{1}{n}\left({\db_j}^2 - 2\db_j\bxmj\right)$. The difference between variance of \db and $\db'$ would then be:

\begin{align*}
    s^2(\db) - s^2(\db') = \frac{1}{n}\left[{\db_j}^2 - {\db'_j}^2 - 2\left(\db_j - \db'_j \right)\bxmj\right]
\end{align*}

Now consider the partial derivative of the variance function with respect to the jth observation:

\begin{align*}
    \frac{\partial{s^2}}{\partial{\db_j}} & = \frac{2}{n}\left(\db_j - \bxmj\right)\\
    \frac{\partial^2{s^2}}{\partial{\db_j}^2} & = \frac{2}{n} > 0
\end{align*}

And thus among all possible values $\db_j$ can take, the variance $s^2$ is minimized when $\db_j = \bxmj$ and maximized when $\db_j$ is equal to the bound that's farthest from $\bxmj$.  \\
Thus the sensitivity bound of variance can be found to be:

\begin{align*}
    \Delta s^2 & = \max_{\db_j,\bxmj} \left[s^2(\db_j)-s^2(\db'_j)\right] \\
    & = \max_{\db_j,\bxmj} \frac{1}{n}\left[{\db_j}^2 - {\db'_j}^2 - 2\left(\db_j - \db'_j \right)\bxmj\right] \\
    & = \max_{\db_j,\bxmj} \left[\frac{1}{n}({\db_j}^2-{\bxmj}^2 - 2(\db_j - \bxmj)\bxmj)\right]\\
    & = \max_{\db_j,\bxmj} \left[\frac{1}{n}({\db_j}^2 - 2\db_j\bxmj + {\bxmj}^2)\right]\\
    & = \max_{\db_j, \bxmj}\left[\frac{1}{n}\left(\db_j - \bxmj\right)^2\right] \\
    & = \frac{1}{n}(\xmax - \xmin)^2
\end{align*}
\end{proof}

\begin{theorem}
Bounding Sensitivity of $f(\db) = \sum_{i=1}^n |\db_i - \bar \db|$.
\end{theorem}

\begin{proof}
Consider an arbitrary change in a particular database value $\db_j$, to a new value $\db_j'$. Setting $d = \db_j' - \db_j$, we notice that the altered mean, $\bar \db' = \bar \db + \frac{d}{n}$. Now, we bound the sensitivity as follows,

\begin{align*}
\Delta f &= \max_{\db, \db'\text{ neighbors}}| f(\db) - f(\db')| \\
&= \max_{\db,\db'\text{ neighbors}}\bigg|\sum_{i \neq j} |\db_i - \bar \db| - \sum_{i \neq j} |\db_i - \bar \db'|\\
&\quad+ |\db_j - \bar \db| - |\db_j' - \bar \db'|\bigg| \\
&\leq \max_{\db,\db'\text{ neighbors}} \bigg( \bigg | \sum_{i \neq j} |\db_i - \bar \db| - \sum_{i \neq j} |\db_i - \bar \db'| \bigg |  \\
&\quad+ \left ||\db_j - \bar \db| - |\db_j' - \bar \db'| \right| \bigg)
\end{align*}

Considering the cases $\sum_{i \neq j} |\db_i - \bar \db| - \sum_{i \neq j}|\db_i - \bar \db'|$ and $|\db_j - \bar \db| - |\db_j' - \bar \db'|$ separately, we have
\begin{align*}
\sum_{i \neq j}& |\db_i - \bar \db| - \sum_{i \neq j}|\db_i - \bar \db'|\\
&=\sum_{i \neq j}|\db_i - \bar \db| - \sum_{i \neq j}|\db_i - (\bar \db + \frac{d}{n})| \\
&\leq  \sum_{i \neq j} |\db_i - \bar \db| - \sum_{i \neq j}|\db_i - \bar \db| + \sum_{i \neq j}| \frac{d}{n}| \\
&=| \frac{d(n-1)}{n}|,
\end{align*}

and

\begin{align*}
|\db_j - \bar \db| - &|\db_j' - \bar \db'|\\
&= |\db_j - \bar \db| - |\db_j+ d - (\bar \db + \frac{d}{n})| \\
&\leq |\db_j - \bar \db| - |\db_j - \bar \db| + |d| + |\frac{d}{n}| \\
&= |d| + |\frac{d}{n}|.
\end{align*}

So, putting the two together, we have
\begin{align*}
\Delta f &\leq\max_{d}\left(|d| + |\frac{d}{n}| + |\frac{d(n-1)}{n}|\right)\\
&=\left(1+\frac{1}{n}+\frac{n-1}{n}\right)\max_{d}|d|\\
&=2\max_{d}|d|.
\end{align*}
Since $|d| \leq (\xmax-\xmin)$, $\Delta f \leq 2(\xmax-\xmin)$.
\end{proof}

\section{Details on D'Orazio and Honaker's algorithm}\label{sec:ORA}

D'Orazio, Honaker, and King \cite{DOrazio2015} describe a method for calculating the standard error for a private estimate of the difference in means between two normally distributed random variables. This is a different case than we are considering.  However, the difference between two normally distributed random variables is itself normally distributed, so their method can be adopted easily enough to our case. Making this change does mean we had to slightly adapt some aspects of the algorithm, and for this reason we have reproduced the exact algorithm we used below.

A few changes and implementation details of note:
\begin{enumerate}
\item Although the algorithm used \expmed to estimate the first and third quartiles of the subsamlped estimates of standard error, they did not give any details about how the upper bound given to \expmed was determined. We felt that since this was a bound on standard error rather than standard deviation, it would need to depend on the size of the database provided. To get around this, we pass a bound for the actual standard deviation, $sd_{\text{max}}$ and divide this by $\sqrt{n}$ to place a bound on the true standard error, $se_{\text{max}}$. However, this bound is likely not conservative enough because the standard errors will follow their own sampling distribution based on the size of the subsets. To address this, we add two standard deviations of the standard error calculated on the $M$ subsets to $se_{\text{max}}$. The standard error of the standard deviation calculated on each subsetis is approximately $\frac{sd_{\text{max}}}{\sqrt{2 \frac{n}{M}}}$. After rearranging terms and scaling by $\frac{1}{\sqrt{n}}$, we get that the standard error on our estimates of standard error is $\frac{sd_{\text{max}}\cdot \sqrt{M}}{n\sqrt{2}}$. We then added three times this value to our bound on the true standard error to get the value we pass to \expmed.
\item The paper also did not give any discussion of how to select the number of subsets on which to calculate the standard deviation. Smith 2011 \cite{Smith2011} gave a heuristic of $\sqrt{n}$ as the number of subsets for a similar algorithm, but we found this to give too few groups. Instead, we empirically optimized the group size at various levels of $n$ and interpolated to approximate the best subsample size for a given database. We found that for all the sample sizes we tried, the best results occurred when the size of each subsample was $2$.

\item In the original paper, the Laplace noise added to the winzorized mean has scale parameter $\frac{|u-l|}{2\eps M}$. We believe this is an error and that the value should be $\frac{2|u-l|}{\eps M}$. This is because half of the $\eps$ budget is consumed by the quartile estimates and so the sensitivity $\frac{|u - l|}{M}$ should be divided by the remaining $\frac{\eps}{2}$.
\end{enumerate}

\begin{algorithm}[H]
  \caption{Construct D'Orazio Mean and SD, \sf ORA}\label{ORA}
  \begin{algorithmic}[1]
  \REQUIRE $\db, \eps, M,  x_{min}, x_{max}, sd_{max}$
  \STATE $\tilde{\db} \gets \bar{x} + L_1$, where $L_1 \sim \Lap \left(\frac{2(x_{max}-x_{min})}{\eps n}\right)$
  \STATE $se_{max} \gets \frac{sd_{max}}{\sqrt{n}}$
  \STATE Divide dataset $\db$ into M subsets $m_1, m_2, ..., m_M$.
  \FOR{$i \gets 1, M$}
      \STATE $s_i \gets \frac{sd(m_i)}{\sqrt{n}}$
  \ENDFOR
  \STATE $S \gets s_1, s_2, ...  s_M$
  \STATE $a \gets \expmed\left(S, \frac{1}{4}, \frac{\eps}{4}, 0, se_{max} + 2 \cdot \frac{sd_{\text{max}}\cdot \sqrt{M}}{\sqrt{2\cdot n^2}}\right)$
  \STATE $b \gets \expmed\left(S, \frac{3}{4}, \frac{\eps}{4}, 0, se_{max} + 2 \cdot \frac{sd_{\text{max}}\cdot \sqrt{M}}{\sqrt{2\cdot n^2}}\right)$
  \STATE $\mu \gets \frac{a+b}{2}$
  \STATE $IQR \gets|a-b|$
  \STATE $u \gets \mu + 2IQR$
  \STATE $l \gets \mu - 2IQR$
  \FOR{$i \gets 1, M$}
      \STATE $s_i \gets \begin{cases}
      u & \text{if } s_i > u \\
      s_i & \text{if }  l < s_i < u \\
      l & \text{if } s_i < l
      \end{cases}$
  \ENDFOR
  \STATE $w \gets \frac{1}{M}\sum^M_{i = 1}s_i$
  \STATE $\tilde{s} \gets w + L_2, L_2 \sim \Lap \left(\frac{2|u-l|}{\eps M}\right)$
  \ENSURE $\tilde{\db}, \tilde{s}$
  \end{algorithmic}
\end{algorithm}

\section{Detailed experimental results}\label{sec:gphs}

\begin{figure*}[h]
    \centering
    \includegraphics[width = \textwidth]{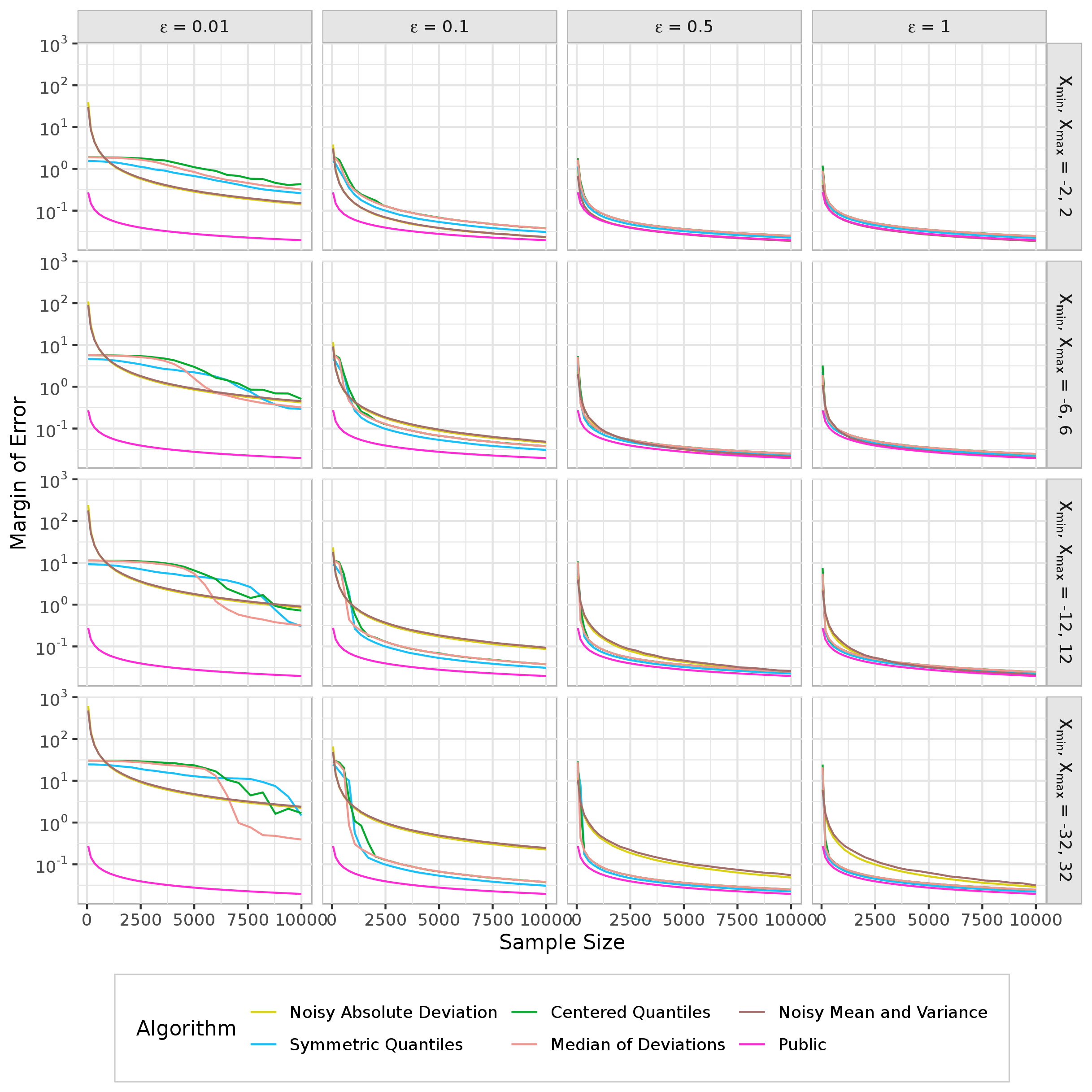} \par
    \caption{Comparison of our algorithms with respect to their average \moe at various database sizes, \eps\ values , and ranges. }
    \label{fig:ours_moe}
\end{figure*}

\begin{figure*}[h]
    \centering
    \includegraphics[width = \textwidth]{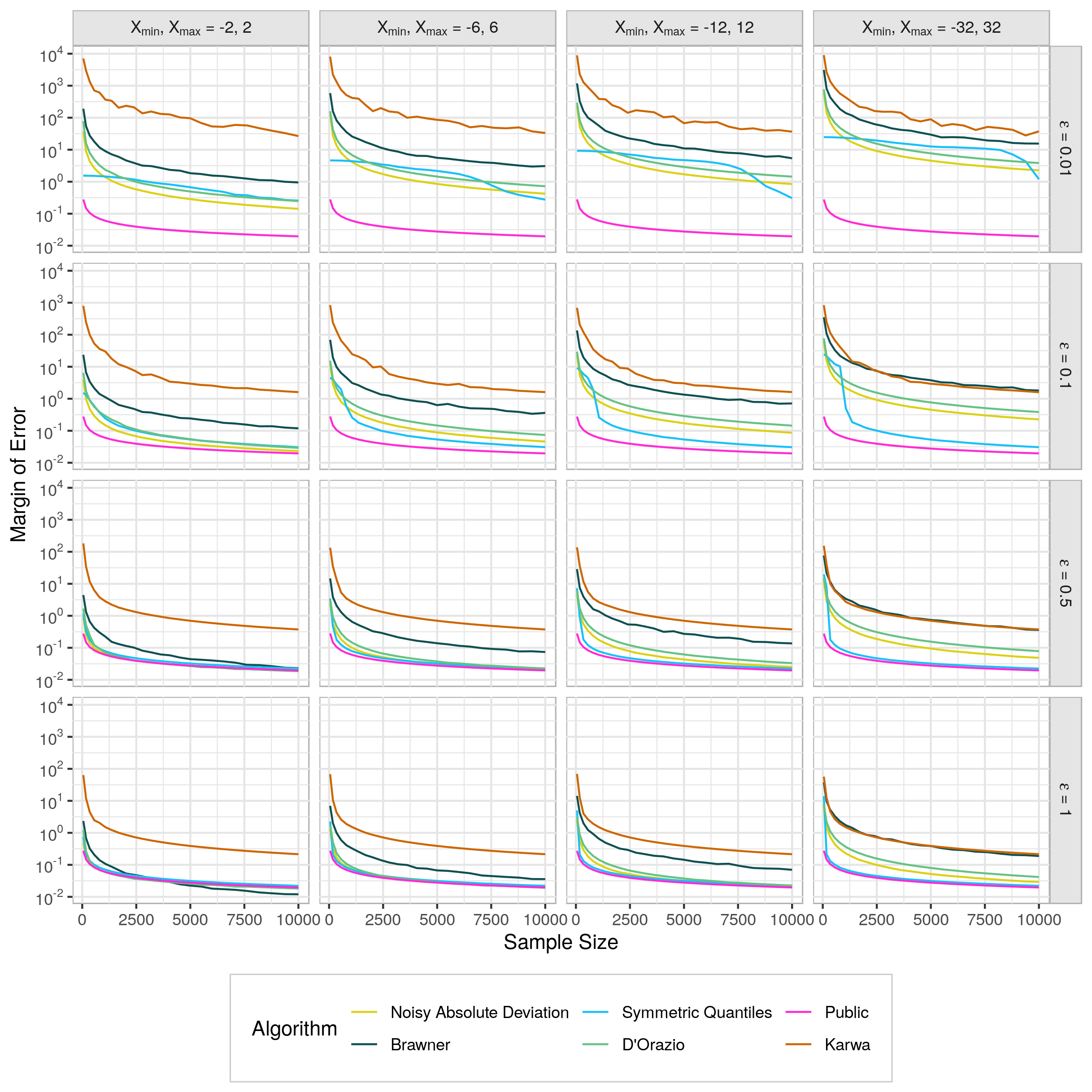} \par
    \caption{Our best algorithms compared to prior work with respect to their average \moe at various databases sizes, \eps \ values and ranges.}
    \label{fig:comp_moe}
\end{figure*}

\begin{figure*}[h]
   \centering
   \includegraphics[width = \textwidth]{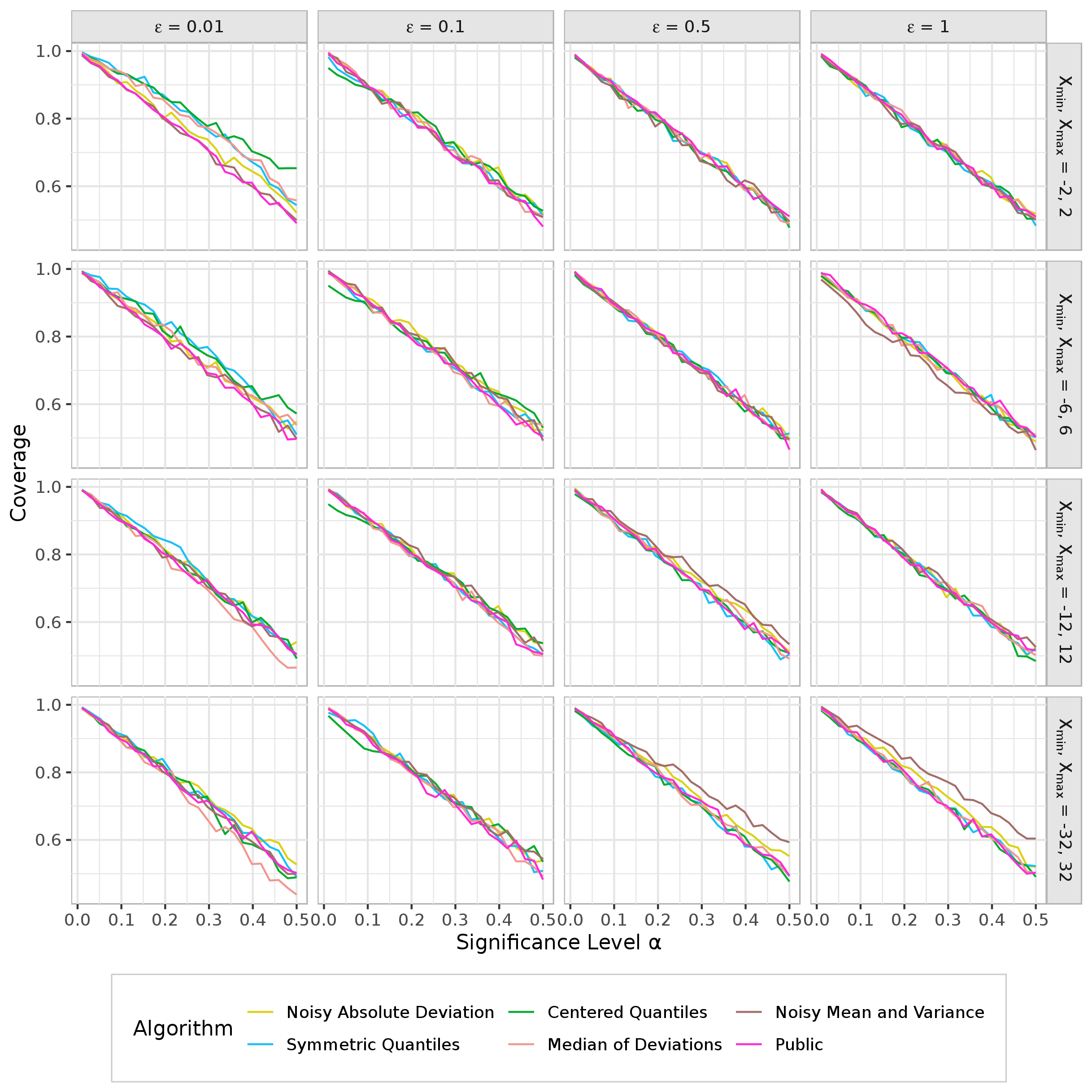} \par
   \caption{The coverage of our algorithms at various significance levels, \eps \ values and ranges.  $n=1000$}
   \label{fig:ours_cov_range_e}
\end{figure*}

\begin{figure*}[h]
   \centering
   \includegraphics[width = \textwidth]{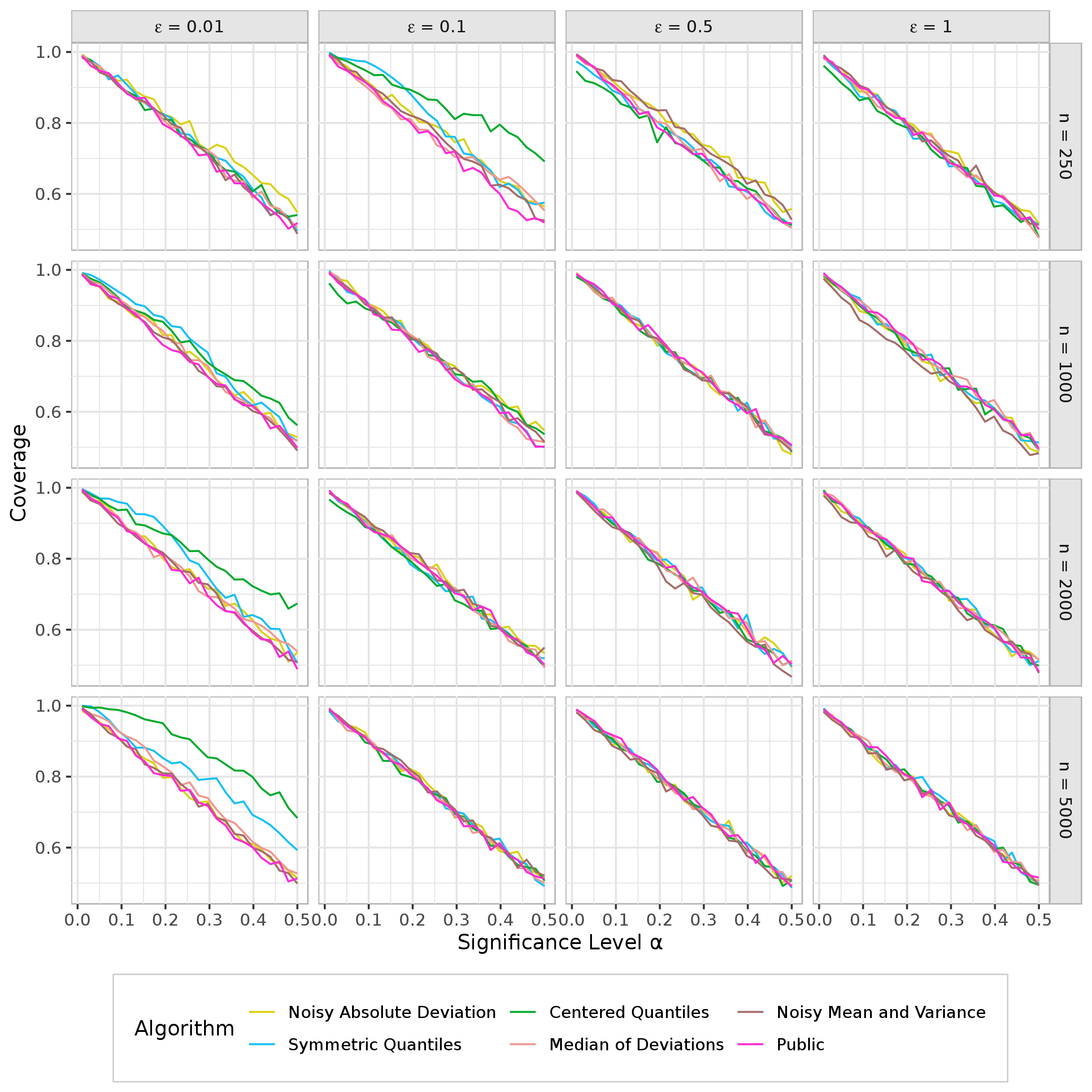} \par
   \caption{The coverage of our algorithms at various significance levels, \eps \ values and sample sizes. }
   \label{fig:ours_cov_n_e}
\end{figure*}

\begin{figure*}[h]
   \centering
   \includegraphics[width = \textwidth]{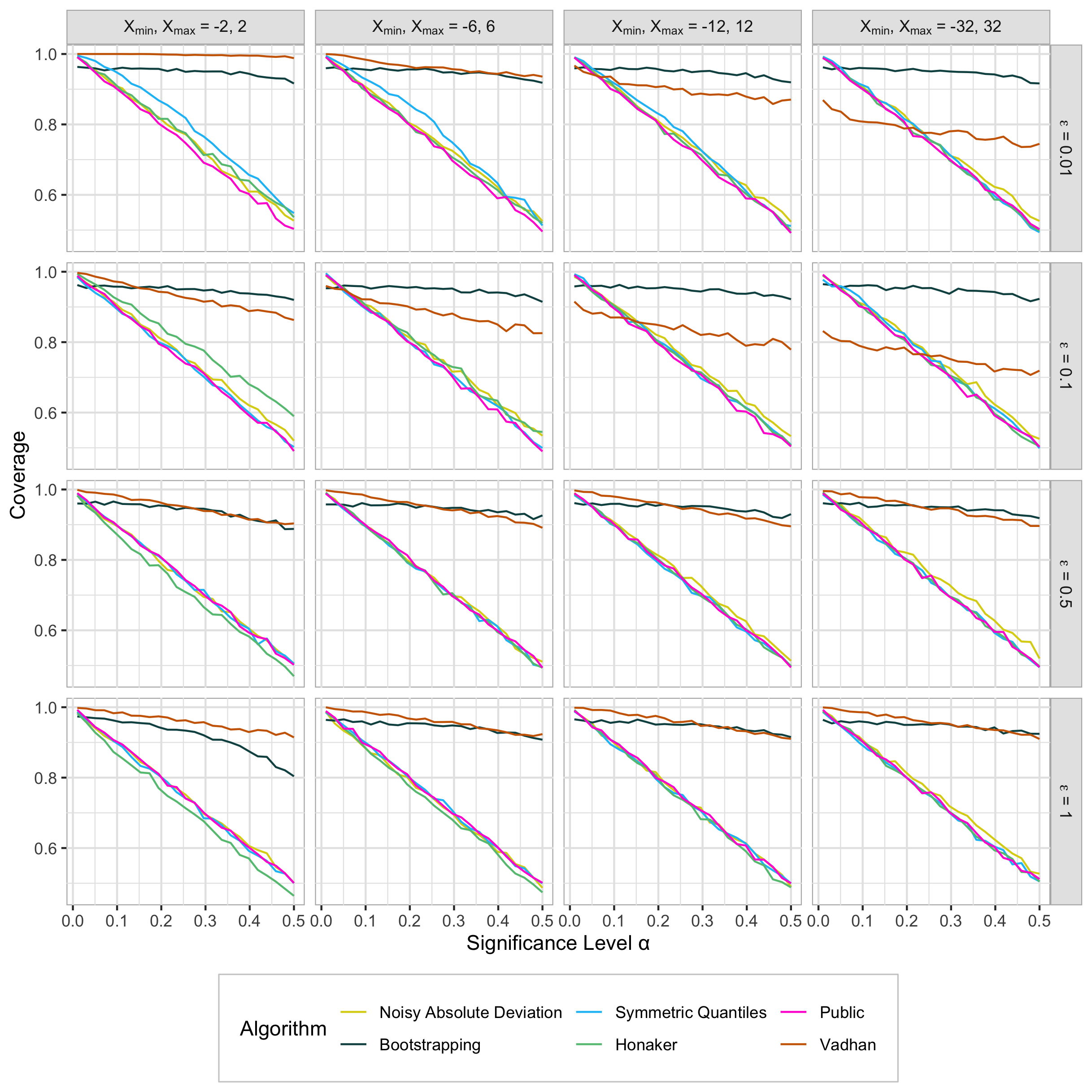} \par
   \caption{Our best algorithms compared to prior work with respect to their coverage at various significance level, \eps \ values and ranges.}
   \label{fig:comp_cov}
\end{figure*}

\begin{figure*}[h]
   \centering
   \includegraphics[width = \textwidth]{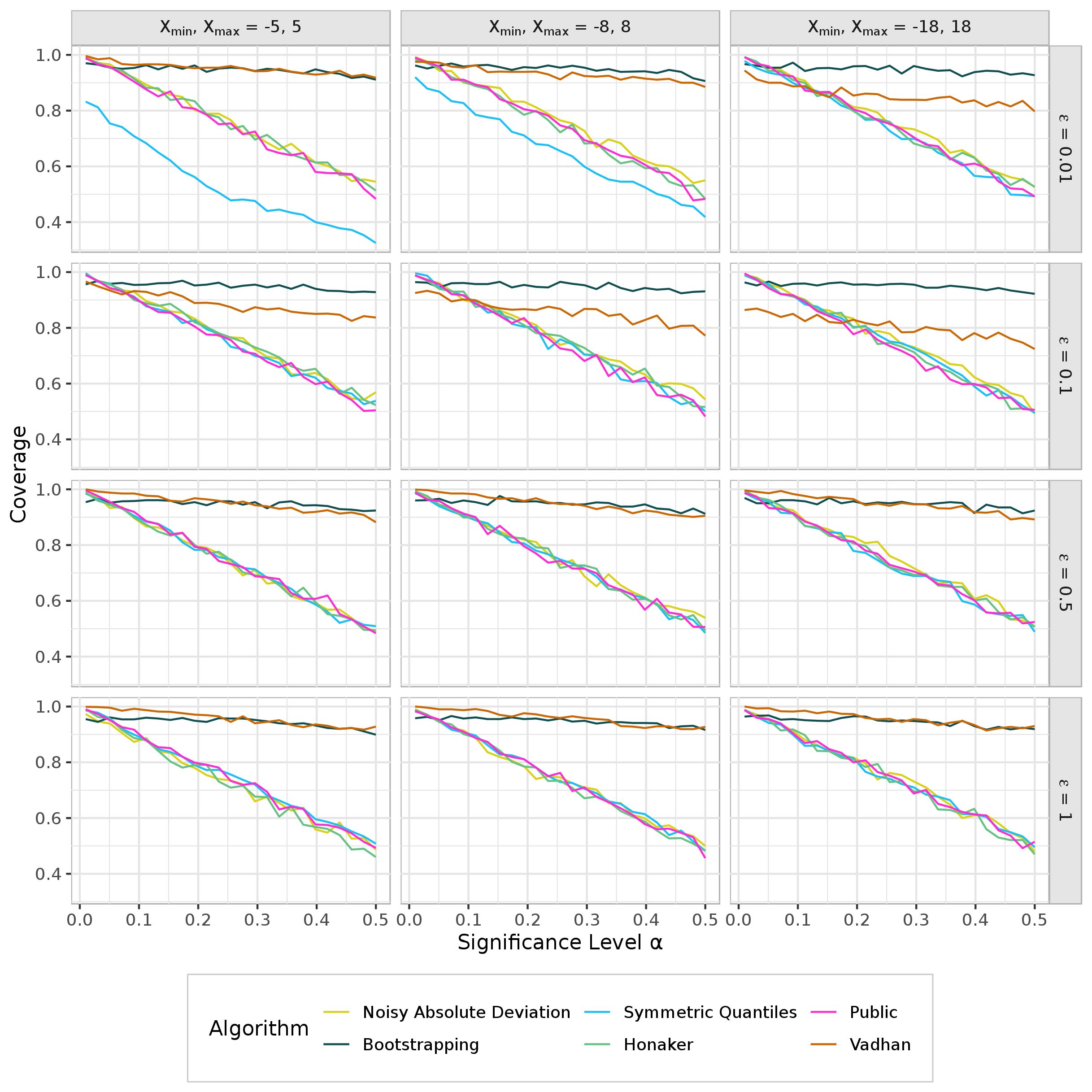} \par
   \caption{The coverage of our algorithms at various significance levels, \eps \ values and sample sizes. Here the true mean is 3 to test an off-center case.}
   \label{fig:comp_non0center_cov}
\end{figure*}

\clearpage
\section{Parameter Allocation and Optimization}\label{sec:alloc}

\begin{figure*}[h]
   \centering
   \includegraphics[width = 0.85\textwidth]{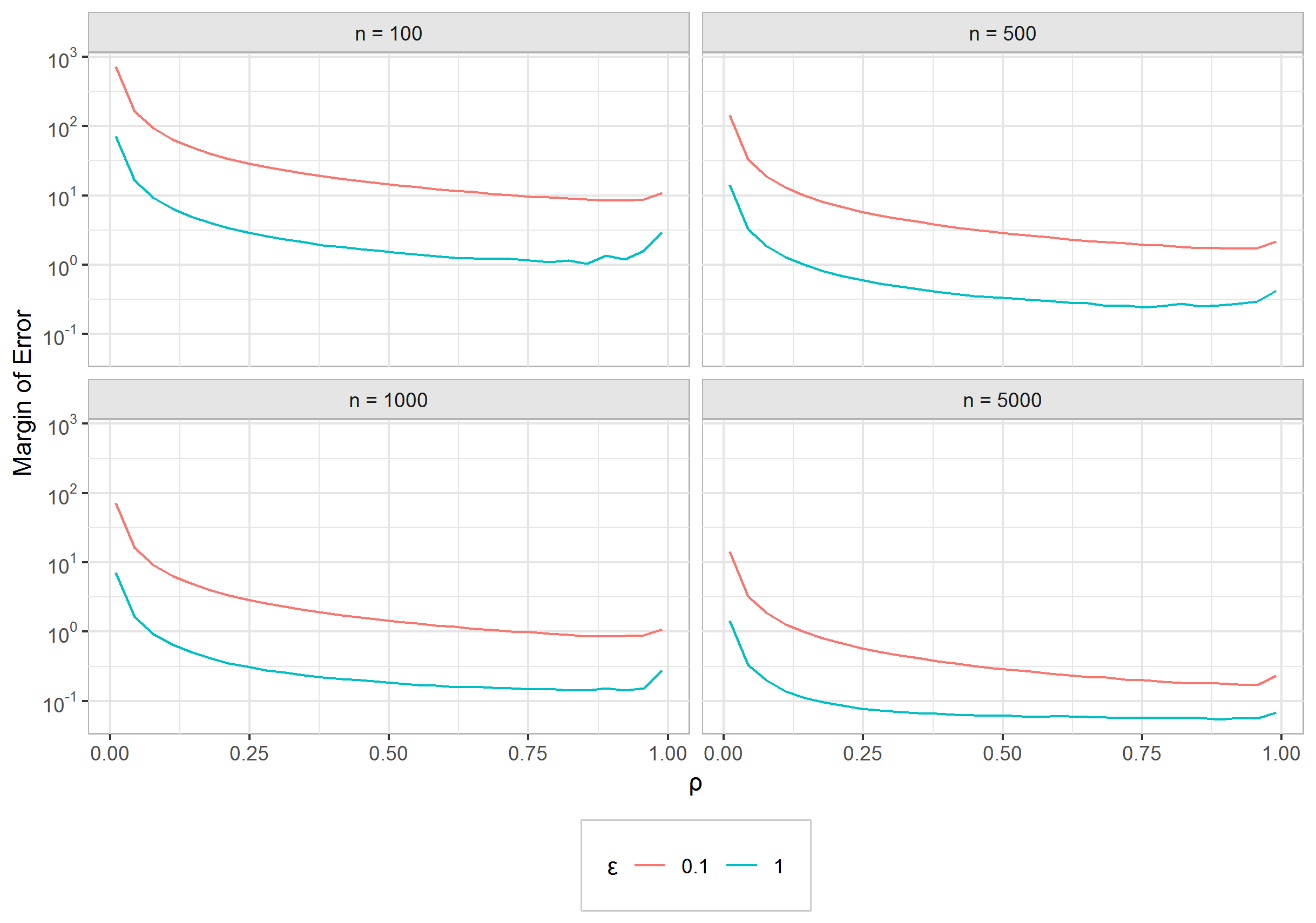} \par
   \caption{The performance of \cilap at various database sizes varying \eps-allocation. \cilap tends to perform the best when $\rho$ is in the range of approximately (0.75, 0.85), where the \moe of confidence intervals are minimized for all database sizes. We thus choose $\rho$ = 0.8 as the optimized \eps-allocation for \cilap.}
   \label{fig:lap_params}
\end{figure*}

\begin{figure*}[h]
   \centering
   \includegraphics[width = 0.85\textwidth]{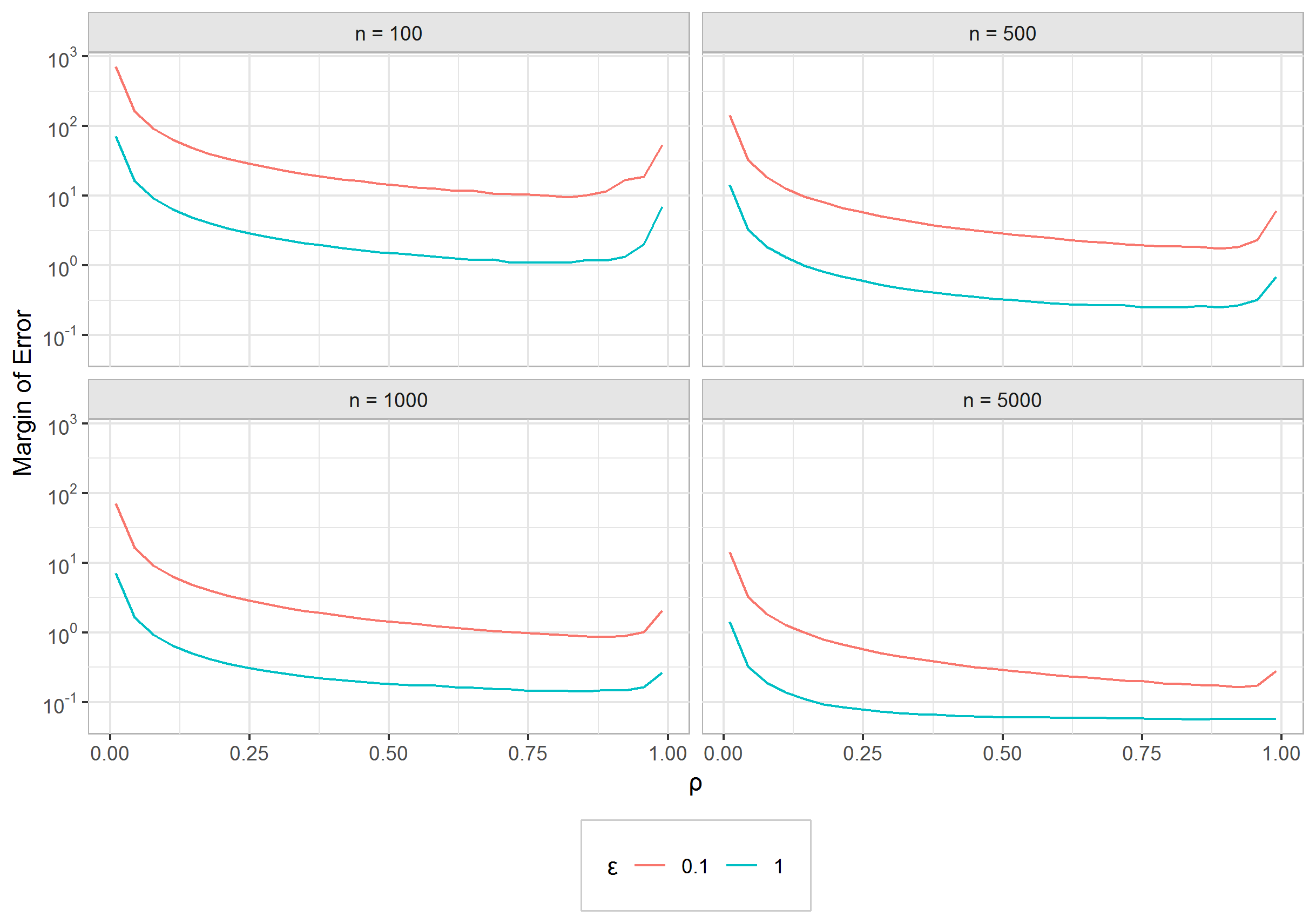} \par
   \caption{The performance of \abslap at various database sizes varying \eps-allocation. \abslap tends to perform the best when $\rho$ is in the range of approximately (0.75, 0.88), where the \moe of confidence intervals are minimized for all database sizes. We thus choose $\rho$ = 0.85 as the optimized \eps-allocation for \abslap.}
   \label{fig:abs_params}
\end{figure*}

\begin{figure*}[h]
   \centering
   \includegraphics[width = 0.85\textwidth]{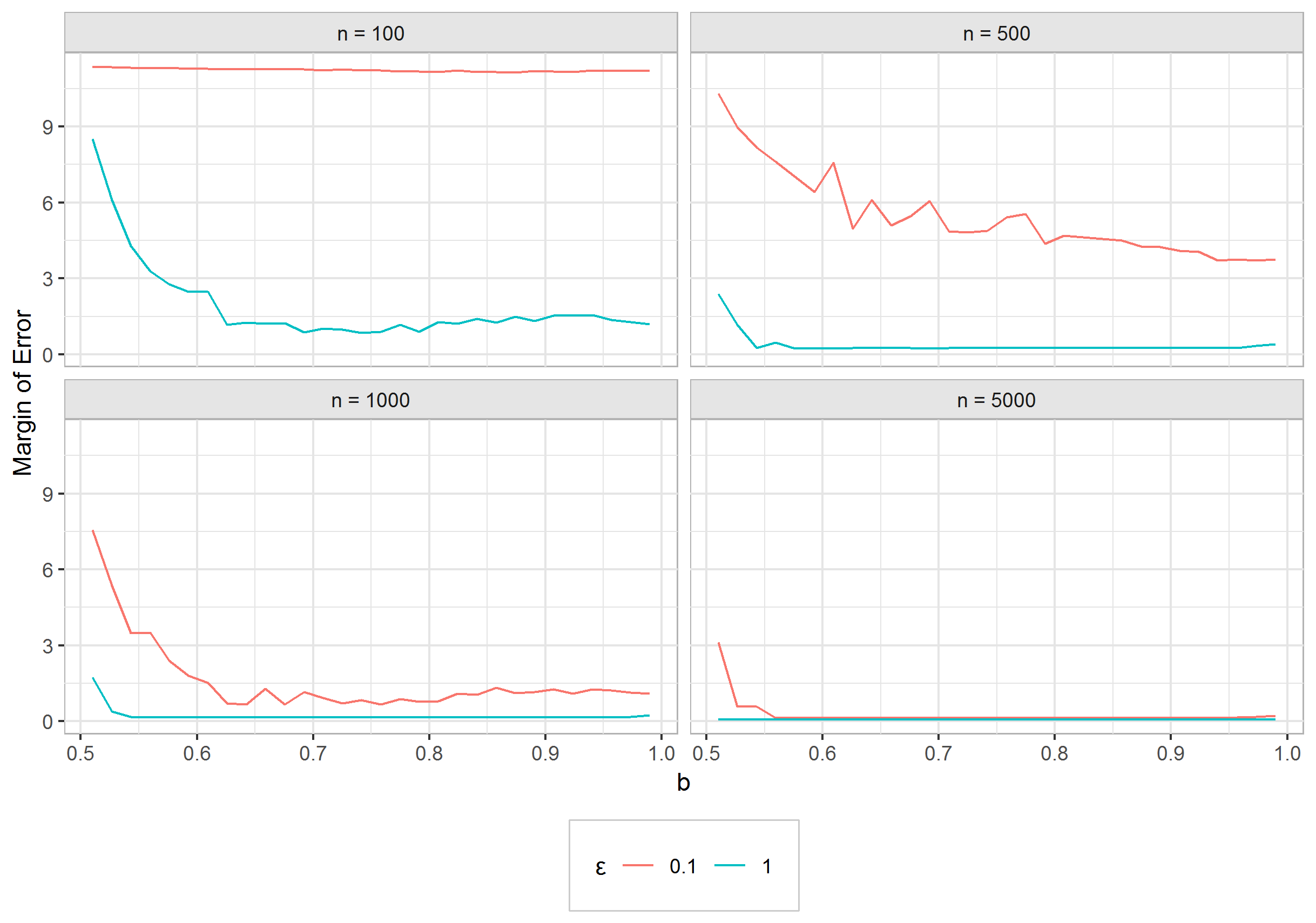} \par
   \caption{The performance of \cexp at various database sizes varying b, the percentile used for generating private measure of spread, with \eps-allocation $\rho$ set to be 0.5. \cexp tends to perform better when b falls into the range of approximately (0.65, 0.8), where the \moe of confidence intervals tend to be relatively small for all database sizes. We here choose b = 0.65 as the optimized value of b for \cexp.}
   \label{fig:median_params}
\end{figure*}

\begin{figure*}[h]
   \centering
   \includegraphics[width = 0.85\textwidth]{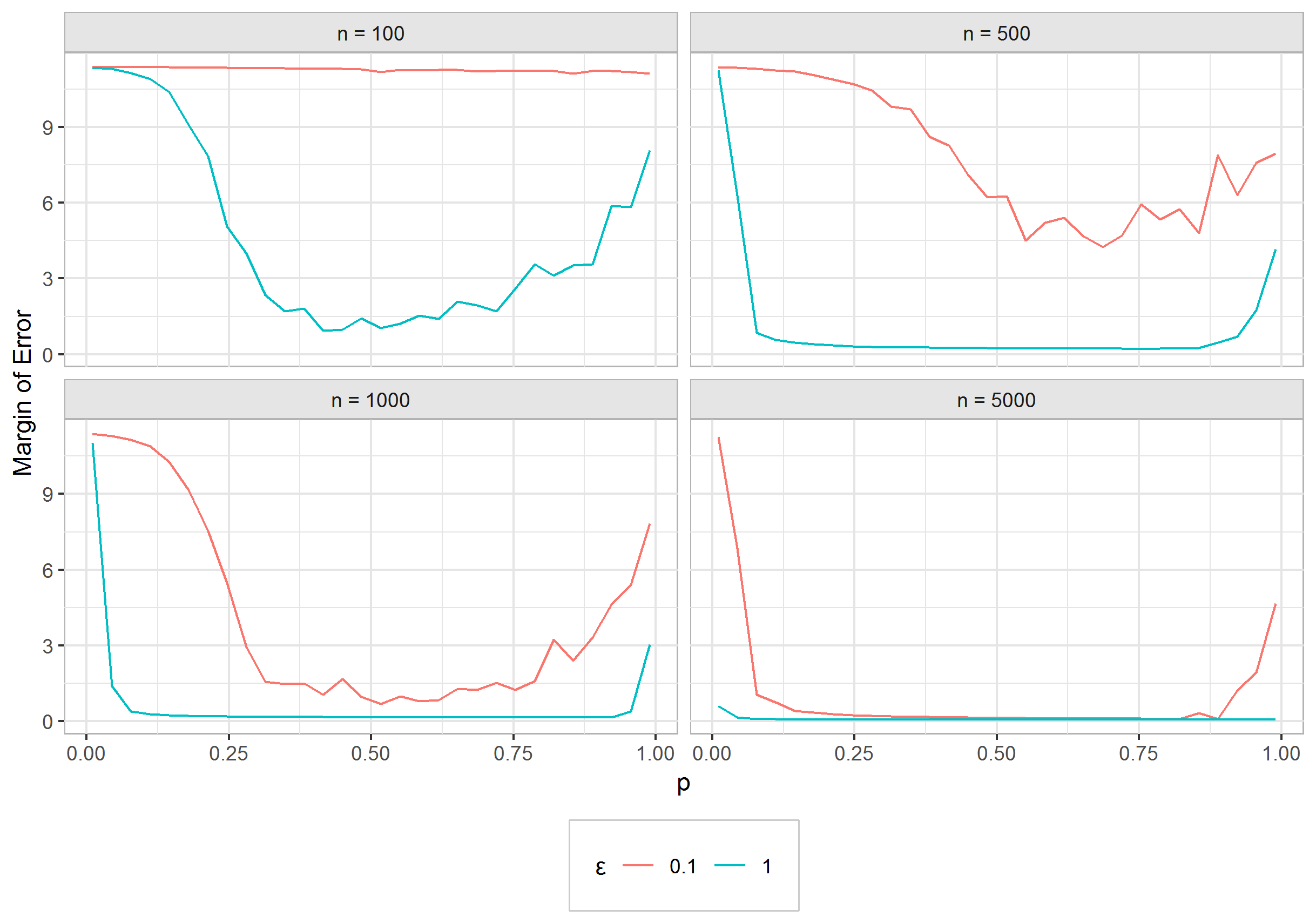} \par
   \caption{The performance of \cexp at various database sizes varying \eps-allocation, with b set to be 0.65. \cexp tends to perform better with $\rho$ in the range of approximately (0.5, 0.75), where the \moe of the confidence intervals are small for all database sizes. We here choose $\rho$ = 0.5 as the optimized \eps-allocation. }
   \label{fig:median_params_p}
\end{figure*}

\begin{figure*}[h]
   \centering
   \includegraphics[width = 0.85\textwidth]{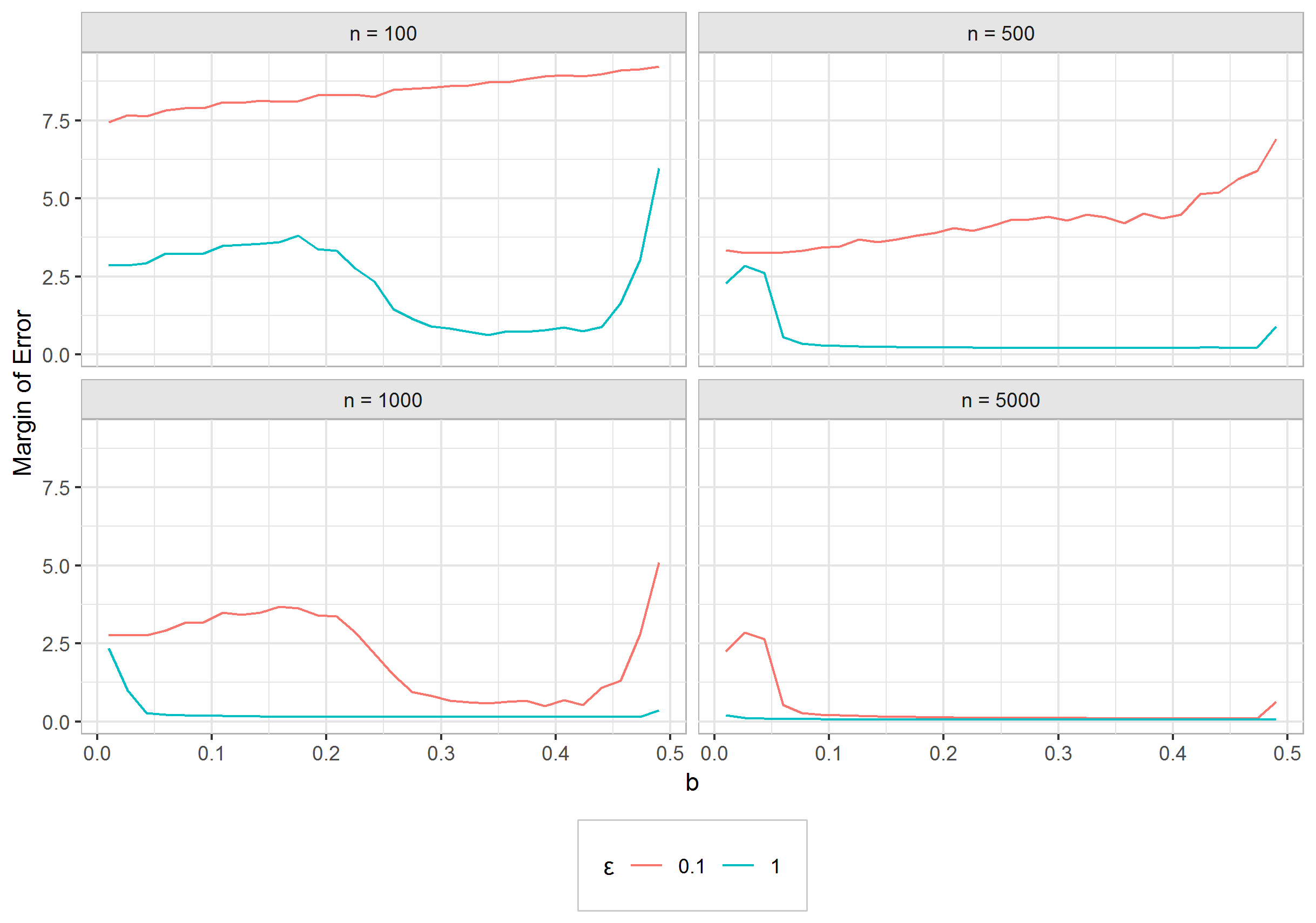} \par
   \caption{The performance of \cdbl at various database sizes varying b, the percentile used for generating private measure of spread. \cdbl tends to perform better when b falls into the range of approximately (0.3, 0.45), where the \moe of confidence intervals are minimized for all database sizes. We here choose b = 0.35 as the optimized value of b for \cdbl.}
   \label{fig:double_params}
\end{figure*}

\begin{figure*}[h]
   \centering
   \includegraphics[width = 0.85\textwidth]{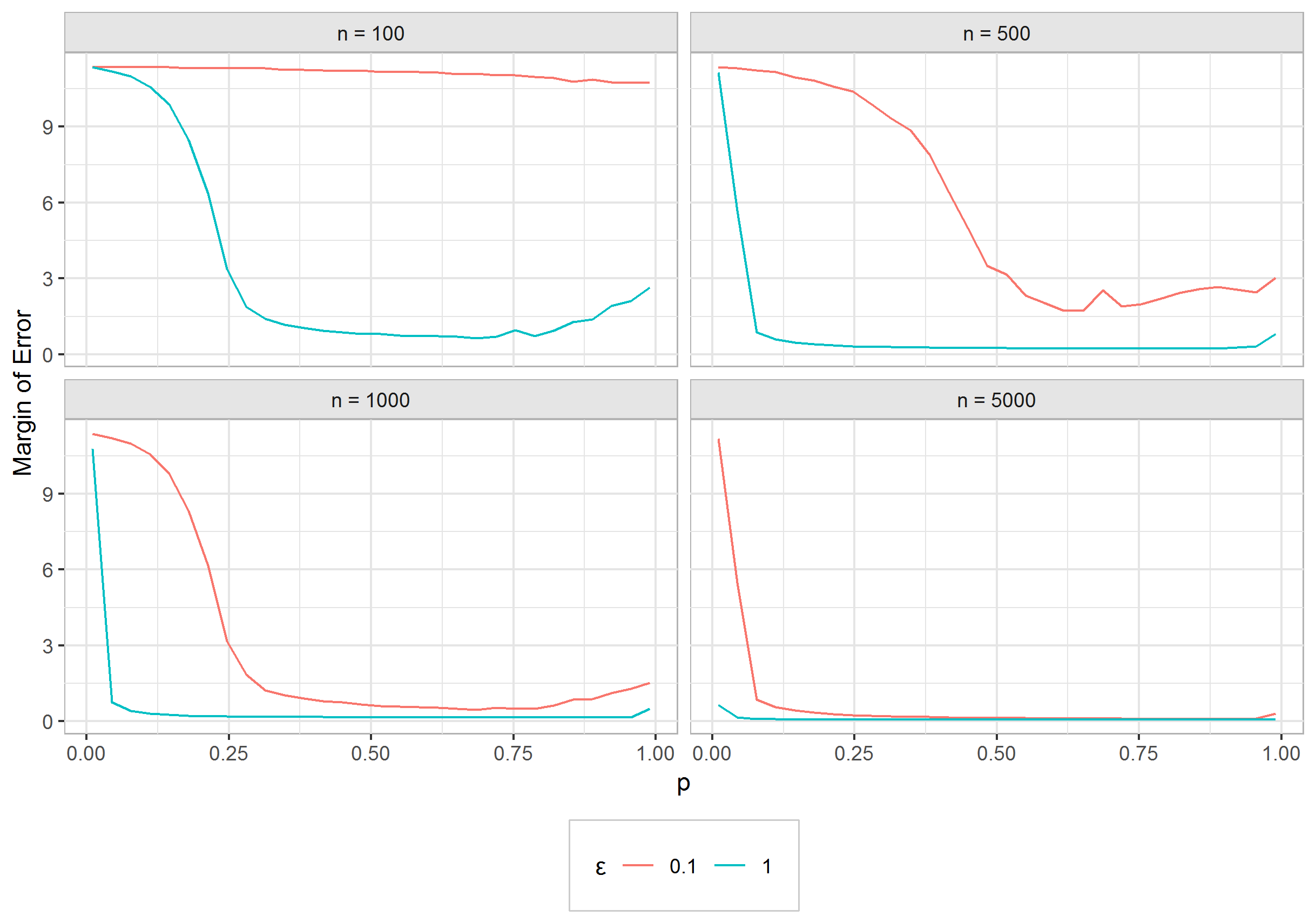} \par
   \caption{The performance of \mad at various database sizes varying \eps-allocation. \mad tends to perform better when $\rho$ falls into the range of approximately (0.45, 0.55), where the \moe of confidence intervals are minimized for all database sizes. We here choose $\rho$ = 0.5 as the optimized value of b for \mad.}
   \label{fig:mod_params}
\end{figure*}

\end{document}